\documentclass[12pt]{article}
\usepackage[T1]{fontenc}
\usepackage{ae,aecompl}
\usepackage{color}
\usepackage{amsmath,amsthm,amssymb,enumerate,mathrsfs}
\usepackage{here,float}
\usepackage{subcaption}

\usepackage{amsfonts}
\usepackage{graphicx}
\usepackage{lmodern}
\usepackage[title]{appendix}

\usepackage{chngcntr}
\usepackage{caption}
\usepackage{ae,aecompl}
\usepackage[a4paper, margin=3.3cm]{geometry}
\usepackage[latin1]{inputenc}
\newtheorem{rem}{Remark}
\usepackage{enumerate}
\newtheorem{prop}{Proposition}

\newtheorem{Ass}{Assumption}
\newtheorem{Corollary}{Corollary}

\def\bkR{{\rm I\kern-.17em R}}
\def\QED{\begin{flushright} QED \end{flushright}}

\def \1n{1\hskip -3pt \mbox{N}}

\def \Frac {\displaystyle \frac }
\def \Int {\displaystyle \int }
\def \Sum {\displaystyle \sum }

\begin{document}

\begin{titlepage}
\thispagestyle{empty}

\title{SIR Model with Stochastic Transmission}
\author{Gourieroux, C.,$^{(1)}$ and  Y.,Lu$^{(2)}$ }
\date{(Preliminary version, November 09, 2020)}\maketitle \vspace{9cm}

C. Gourieroux acknowledges the financial support of the chair ACPR : "Regulation and Systemic Risks", and of the Agence Nationale de la Recherche (ANR-COVID), grant ANR-17-EUR-0010. Y. Lu thanks the CNRS and the Labex MME-DII for d\'el\'egation grants. 

\addtocounter{footnote}{1} \footnotetext{University of Toronto, Toulouse School of Economics and CREST. Email: christian.gourieroux@ensae.fr}
\addtocounter{footnote}{1}\footnotetext{Universit\'e Sorbonne Paris Nord, France and Concordia University, Canada. Email: yang.lu@univ-paris13.fr}
%\addtocounter{footnote}{1}\footnotetext{CREST. }

\end{titlepage}
\newpage
\begin{center}
SIR Model with Stochastic Transmission\\
%\\[10mm]
\textbf{Abstract}
\end{center}
\vspace{1em}

The Susceptible-Infected-Recovered (SIR) model is the cornerstone of epidemiological models. However, this specification depends on two parameters only, which implies a lack of flexibility and the difficulty to replicate the volatile reproduction numbers observed in practice. We extend the classic SIR model by introducing nonlinear stochastic transmission, to get a stochastic SIR model. We derive its exact solution and discuss the condition for herd immunity. The stochastic SIR model corresponds to a population of infinite size. When the population size is finite, there is also sampling uncertainty. We propose a state-space framework under which we analyze the relative magnitudes of the observational and stochastic epidemiological uncertainties during the evolution of the epidemic. We also emphasize the lack of robustness of the notion of herd immunity when the SIR model is time discretized.\vspace{1em}

\textbf{Keywords~:} SIR Model, Population-at-Risk, Semi-Parametric Model, Mover-Stayer Model, Herd Immunity, Reproductive Number, Stochastic Transmission, Systemic Risk.

\newpage

\section{Introduction}
The Susceptible-Infected-Recovered (SIR) model, introduced by Kermack, McKendrick (1927), is the cornerstone of epidemiological models. They extend the initial smallpox epidemiological model of Bernoulli (1760) by introducing the three compartments, S, I, R. It is still largely used in practice and has been extended in various directions, generally by including more states such as the possibility of decease, or the medical treatment states [Djogbenou et al. (2020)], or multiple infectious states [Anderson, May (1991)], or by also considering vital dynamics, i.e. the exogenous birth and death processes [see e.g. Allen (1994), Section 5, Harko et al. (2016)], temporary immunity [Xu, Li (2018)], or heterogeneity [Alipoor, Boldea (2020)].\\

The success of the basic SIR model (and of its several extensions) is largely due to the key roles of the notions of reproductive number and herd immunity, that become important indicators for health policy when a new disease (as COVID-19) appears. However, these notions are model based, that is, their interpretation and use strongly depend on our belief on the SIR model. Moreover, it has been observed that daily estimations of the reproductive number are highly volatile, a stylized fact that is not really compatible with the classic SIR model. The aim of this paper is to extend the standard SIR model in two directions: first we allow the transmission of the disease to be captured by a function (transmission function), instead of being summarized by a scalar. This leads to a semi-parametric SIR model able to capture the mover-stayer phenomenon. In a second step this transmission function is made stochastic in order to introduce an infinite dimensional uncertainty in the basic deterministic SIR model. \vspace{1em}

We introduce the Semi-Parametric SIR (SPSIR) model in Section 2. Following the approach in Harko, Lobo, Mak (2016), we derive the exact solution of the SPSIR model. This solution is obtained by solving appropriate integral equations. The forms of these equations are used to discuss the conditions for herd immunity and to derive the final sizes, if applicable. Whereas in the basic SIR model the final sizes (long run parameters) are in a one-to-one relationship with the reproductive number at the beginning of the epidemic (a short run parameter), this undesirable constraint no longer exists in the general case. This means that some standard consequences of the SIR model are not robust to small changes of the transmission function. The transmission function is made stochastic in Section 3 leading to the SIR model with stochastic transmission. As a consequence all standard notions such as the evolution of infected people, of recovered people, the reproductive number(s), the existence of herd immunity and the final sizes become stochastic. Simulations of disease trajectories are provided in Section 4. We first compare simulations of the SIR model with stochastic transmission performed from the continuous time model and its Euler discretization. Then we perform a Monte-Carlo study to analyse the uncertainties on the reproductive ratios, the time and length of the peak of infection and the herd immunity ratio.
%Statistical inference is discussed in Section 5.
The SIR model with stochastic transmission assumes a population of infinite size. In practice, the size of the population is large, but finite. In Section 5, we introduce the model for count observations associated with the SIR model with stochastic transmission. We first rewrite the model as a (nonlinear) state space models with two layers of state equations. This state space representation is used to simulate trajectories of virtual observed counts. This allows for comparing the relative magnitudes of the epidemiological and demographic uncertainties and how they evolve during the progression of the epidemic. %to facilitate the statistical inference.
 Section 6 concludes. Proofs are gathered in Appendices. In particular we compare the continuous time SIR model with its crude time discretized version to show that the conditions for herd immunity are not robust to time discretization.

\section{Semi-Parametric SIR Model}
The aim of this section is to extend the standard SIR model by introducing a functional transmission parameter. We first derive the solution of the SPSIR model. Then we discuss the existence of herd immunity, the expression of the final sizes, and the reproduction number(s).

\subsection{The model}
\setcounter{equation}{0}\def\theequation{2.\arabic{equation}}

A SIR model considers the evolution of proportions with three compartments: S: Susceptible (not infected, not immunized (i.e. the Population-at-Risk)), I: Infected, Infectious (not immunized), R: Recovered (immunized after the infection). The (theoretical) proportions of the population in each compartment at date $t$ are denoted $x(t), y(t), z(t).$ The SPSIR model is a differential system of the type~:

\begin{equation}
  \left\{
  \begin{array}{lcl}
  \Frac{d x(t)}{dt} & = & - a [x(t)] y (t), \\ \\
  \Frac{dy(t)}{dt} & = & a[x(t)] y(t) - c y(t), \\ \\
  \Frac{dz(t)}{dt} & = & c y (t),
  \end{array}
  \right.
\end{equation}

\noindent where function $a$, called the transmission function, characterizes the instantaneous  transmission rate of infection, and $c$ is a constant rate of recovering. Throughout the rest of this paper, we assume that:
\begin{Ass}
 Function $a$ is strictly positive on $]0,1]$, such that $a(0)=0$ (no susceptible, no possibility of infection), and scalar $c$ is strictly positive.
\end{Ass}

Since $x(t) + y(t)+z(t)=1$, the system is such that $\Frac{dx(t)}{dt} + \Frac{dy(t)}{dt} + \Frac{dz(t)}{dt} = 0$. We have:%Moreover, by assumption $\Frac{dx(t)}{dt} \leq 0, \Frac{dz(t)}{dt} \geq 0$. Therefore $x(t)$ [resp. $z(t)$] is a decreasing (resp. increasing) function of time. Since these functions are bounded, we get the following result.\vspace{1em}

\begin{prop} 
	If $x(0) >0, y(0)>0$, a positive solution $[x(t), y(t), z(t)]$ of the SPSIR is such that :
\begin{itemize}
\item $x(t)$ \text{\normalfont [}resp. $z(t)$\text{\normalfont]} is a decreasing \text{\normalfont [}resp. increasing\text{\normalfont ]} function of time.

\item In particular $x(t)$ \text{\normalfont [}resp. $z(t)$\text{\normalfont]} tends to a limit $x(\infty)$ \text{\normalfont [}resp. $z(\infty)$\text{\normalfont]}, when $t$ tends to infinity. $y(t)$ tends to $y(\infty) = 1-x(\infty) - z(\infty)$.

\item Moreover all derivatives $\Frac{dx (t)}{dt}, \Frac{dy (t)}{dt}, \Frac{dz(t)}{dt}$ tend to zero, when $t$ tends to infinity, and, by writing (2.1) for $t=\infty$, we deduce that $y(\infty) = 0.$
\end{itemize}
\end{prop}

\textbf{Example 1~:} The basic SIR model introduced in Kermack, McKendrick (1927) [see also Hethcote (2000), Brauer, Castillo-Chavez (2001))] corresponds to the linear case $a(x) = \alpha x, \alpha >0$.\vspace{1em}

\textbf{Example 2~:} This basic model can be extended by considering power function of $x$ as : $a(x) = \alpha x^\beta, \alpha > 0, \beta > 0.$ This is the power law of mass action initially introduced in Wilson, Worcester (1945), p26 [see also Stroud et al.(2006)]. These models have a frailty representation (see Appendix D).

\begin{rem} 
	The nonlinear incidence rate in the first equation of (2.1) is introduced w.r.t. $x(t)$. A part of the epidemiological literature has also considered nonlinearities w.r.t. $y$ in order to account for crowding effects [see e.g. Li et al. (2015), Fan et al. (2017)], or network heterogeneity [Roy, Pascual (2006)]. We keep the linearity in $y$ to allow for exact solution of the epidemic model.
\end{rem}

\paragraph{A survivor function interpretation for $x(t)$.} Function $x(t)$ is decreasing in time and bounded by 0 and 1, and thus has the interpretation of a survivor function. Indeed, under the assumption that all individuals are exchangeable at the beginning of the epidemic, it is the probability that a given individual is still susceptible at time $t$. This suggests to rewrite the first equation of the differential system (2.1) into:

\begin{equation}
\Frac{d \log x(t)}{dt}= - \frac{a [x(t)]}{x(t)} y (t),
\end{equation}
where $\log x(t)$ is minus the cumulative hazard function, whereas $\frac{a [x(t)]}{x(t)} y (t)$ is the (instantaneous) hazard function. Thus, the benchmark SIR model assumes that the hazard function is independent of $x$, but is proportional in $y$, which corresponds to a Cox proportional hazard model [Cox (1972)] in epidemiological time $x(t)$.   
 
Another important special case is when $\frac{a(x)}{x}$ is increasing in $x$. In other words, as the epidemic advances and $x(t)$ decreases, the partial derivative of the hazard function with respect to $y(t)$ decreases, reflecting the fact that the surviving individuals are \textit{ex post} more robust for a fixed number of contacts. Hence the SPSIR model can capture the effect of vulnerability heterogeneity, that is the frailty effect and the mover-stayer phenomenon [see e.g. Blumen, Kogan, McCarthy (1955), Goodman (1965)]. A decreasing pattern of $\frac{a(x)}{x}$, though less likely in practice, could be potentially motivated by a mutation of the virus which makes it more and more dangerous as the epidemics progresses. \vspace{1em}

\textbf{Example 3~:} The mover-stayer phenomenon can be simply represented by a piecewise transmission function~: $a(x) = \alpha_1 x,\; \mbox{if}\; x>\gamma, = \alpha_2 x, \; \mbox{if}\; x<\gamma$, with $\alpha_2 < \alpha_1$. $\alpha_1$ is the transmission rate for highly vulnerable (frail) individuals, the first to be infected, and $\alpha_2$ for the least vulnerable (frail) individuals.\vspace{1em}

\textbf{Example 2 (cont.)~:} The condition $\frac{a(x)}{x}$ increasing in $x$ is equivalent to the inequality $\beta > 1.$\vspace{1em}

%\textbf{Example 4~:} 

%The second equation in system (2.1) also implies general patterns for the theoretical proportion of infected people.\vspace{1em}

\paragraph{Behavior of $y(t)$. } Let us now study how the behavior of $y(t)$ over time depends on the functional parameter $a(\cdot)$. We have the following proposition: 
\begin{prop}
	%Let us assume that $c<a(1)$ so that $a^{-1} (c)$ is smaller than 1. Then:
%	 Let us assume that $a(.)$ is an increasing function. Then:
	 \begin{itemize}
	  \item If $a^{-1} (c) \in ]0,x(0)[$, the proportion $y(t)$ is first increasing, then decreasing, with a peak at $y = a^{-1} (c).$ Since $x(0) \leq 1$, a necessary condition for this to happen is $a^{-1}(c)<1$, or $c<a(1)$.
\item If $a^{-1} (c) > x(0)$, the proportion $y(t)$ is decreasing. In particular, a sufficient condition for this to happen is  $c>a(1)$.\vspace{1em}
\end{itemize}
\end{prop}
In the first case, the proportion of susceptible people and then the proportion of recovered people (by the third equation) have a sigmoidal pattern, extending the logistic pattern initially introduced in Feller (1940) as a basic law of growth in biology. By considering different nonlinear transmission functions $a$, the SPSIR model allows for different growth models in this class of sigmoidal patterns [see e.g., Richards (1959), Kahi et al. (2003) for such growth functions].\vspace{1em}

%Note, also that since $x(0)<1$, the inequality $a^{-1} (c) \in ]0,x(0)[$ can be satisfied only if $a^{-1}(c)<1$, or equivalently $c<a(1)$. \vspace{1em}

The condition in Proposition 2 depends on the value of transmission function $a$ at $x=1$, that is, close to the early phases of the epidemic. %It has been noted in the literature that this condition has a local interpretation [see e.g.]. 
Indeed, at the beginning of the epidemic, system (2.1) can be linearized in a neighbourhood of $x(0) =1, y(0)=0,$ as~:

\begin{equation}
\label{linearized}
\left\{
\begin{array}{lcl}
\Frac{d x(t)}{dt} & = & - a (1) y (t), \\ \\
\Frac{dy(t)}{dt} & = & [a(1)-c] y(t), \\ \\
\Frac{dz(t)}{dt} & = & c y (t).
\end{array}
\right.
\end{equation}

The linearized dynamics for $y(t)$ is stable, iff $a(1) <c$. Therefore the condition in Proposition 2 is the same as the stability condition of this linearized version at the beginning of the epidemic. Intuitively, under this stability condition, we cannot escape far away from the initial state (see a more detailed discussion in Section 2.4, Proposition 5).

\subsection{Solution to the SPSIR Model}

The SPSIR model has a closed form solution that can be derived from the initial condition $(x(0), y(0), z(0)),$ following the approach in Harko, Lobo, Mak (2016).\vspace{1em}

\begin{prop}[Solution for $x(t)$] Let us assume $y(0)>0,$ and the transmission function $a$ strictly positive on $]0,1]$. %and $1/a(x)$ is non integrable for $x=0$; 
Then $x(t)$ is the solution of the integral equation~:

\begin{equation}
\label{integralequation}
  t = \Int^{x(t)}_{x(0)} \Frac{dv}{a(v) \{ v-x(0) - y(0) - c \Int^v_{x(0)} \Frac{du}{a(u)}\}}.
\end{equation}
\end{prop}
\begin{proof}
	See Appendix A.
\end{proof}
The integration on the right hand side is with respect to $v$, which is a proportion, instead of time $t$. This can be understood from the change of time used in the proof. The calendar time has no epidemiological interpretation and has to be replaced by the epidemiological, intrinsic time $x(t)$ that accounts for the speed of transmission (see also the above interpretation of $x(t)$ in terms of survivor function). \\

\textbf{Example 1 (cont.)~:} If $a(x) = \alpha x,$ the integral equation becomes~:

\begin{equation}
  t = \Int^{x(t)}_{x(0)} \Frac{dv}{v\{ \alpha v - \alpha [x(0) + y(0)] - c [\log v -\log x(0)]\}}.
\end{equation}

This corresponds to eq. (26) in Harko et al. (2016).\\

\textbf{Example 2 (cont.)~:} If $a(x) = \alpha x^\beta, \alpha > 0, \beta > 1$ (to ensure the non integrability of ${1}/{a(x)}$ at $0$), the integral equation becomes~:

\begin{equation}
  t = \Int^{x(t)}_{x(0)} \Frac{dv}{v^\beta \{ \alpha v - \alpha [x(0) + y(0)] - \Frac{c}{ (1-\beta)} [v^{1-\beta} - x(0)^{1-\beta}]\}}.
\end{equation}
\vspace{1em}

The solution for $y(t)$ is easily deduced. Indeed, from system (2.1), we get~: $\Frac{dy(t)}{dx(t)} = -1 + \Frac{c}{a[x(t)]}$, and by integration, we deduce the following Corollary: 

\begin{Corollary}[Solution for $y(t)$] Under the conditions of Proposition 3, we get~:
\begin{equation}
\label{fromytox}
  y(t) - y(0) = -x(t) + x(0) + c \Int^{x(t)}_{x(0)} \Frac{du}{a(u)}.
\end{equation}
\end{Corollary}
Finally, $z(t)$ is deduced as $z(t) = 1-x(t) - y(t).$
\subsection{Herd Immunity}

From Proposition 1, two cases can arise in the long run~: either $x(\infty) =0, y(\infty) =0, z(\infty) =1,$ or $x(\infty) >0, y(\infty)=0, z(\infty)<1.$ In the first case, all individuals will be infected and then immunized. In the second case, the infection stops when the (theoretical) proportion of immunized people is sufficiently large. Then $z(\infty)$ is the ratio of herd (or population) immunity and $x(\infty)$ is the final size of the susceptible population.
\\

To characterize the two cases, we have to consider the domain of existence of the integral equation (2.4) and more precisely the poles of the denominator of the integrand.
\\

$i)$ There might be at least one pole which is $v=0$, under the extra assumption that $1/a(x)$ is non integrable at zero. Indeed, in this case, within a neighbourhood of $v=0$, the primitive in (2.2) is equivalent to $\Int^v_{c} \Frac{du}{ca(u) \int_{x(0)}^v \Frac{du}{a(u)}} \simeq \Frac{1}{c} \log [\int_{x(0)}^v \Frac{du}{a(u)}],$ that is close to infinity.
\\

$ii)$ But the non integrability condition of $1/a(x)$ at 0 is not necessary, since the denominator in the integrand can be also equal to zero for some positive value $v^*$, which should necessarily be a root of the following equation~:

\begin{equation}
\label{theequation}
  v^* - x(0) - y(0) - c \Int^{v^*}_{x(0)} \Frac{du}{a(u)} = 0, \;\mbox{with}\; 0<v^* < x(0).
\end{equation}

%More precisely, we have:

%\begin{Corollary} If there exist solutions of eq. (2.6) strictly between 0 and $x(0)<1$, then $x (\infty) = v^*$, where $v^*$ is the largest such solution. In this case, we will say that the population will reach \textbf{herd immunity}. Moreover, in this case, the herd immunity is only reached when time $t$ goes to infinity, that is, the solution $x(t)$ decreases from $x(0)$ at $t=0$ to $v^*$ at infinity without ever reaching zero.  
	
%Otherwise, $x(\infty) = 0.$
%\end{Corollary}

%Since $x(0) > x(t) > x(\infty)$, this implies in particular that the solution $x(t)$ is always between $0$ and $1$, and in particular strictly positive if $x(0) > 0.$\vspace{1em}

It remains to derive a simple condition ensuring the existence of such a root $v^*$. We have:
 
\begin{prop} If $a(x)$ is increasing in $x$, and $x(0)>0, y(0)>0$, then
	\begin{enumerate}[$i)$]
		\item 
 there is one and only one solution $v^*$ to equation \eqref{theequation} lying strictly between 0 and $x(0)<1$, if and only if:
	\begin{equation}
	\label{herdimmunitycondition}
	c\int_0^{x(0)} \frac{\mathrm{d}u}{a(u)}> x(0) + y(0).
	\end{equation}
	In this case, the population will reach \textbf{herd immunity}, with $x (\infty) = v^*$.
	\item Moreover, the herd immunity is only reached when time $t$ goes to infinity, that is, the solution $x(t)$ decreases from $x(0)$ at $t=0$ to $v^*$ without ever reaching $v^*$ for finite $t$.  
	
\item If inequality \eqref{herdimmunitycondition} does not hold, 
then $x(t)$ goes to zero in finite time, that is, there exists $T_\infty$ such that $x(t)=0$ for all $t \geq T_\infty$, and, after this extinction time, $y(t)$ decreases at an exponential rate with: $\Frac{dy(t)}{dt}  = - c y(t)$. 
 	\end{enumerate}
 \end{prop}
 \begin{proof}
 	See Appendix B.1.
 \end{proof}

Proposition 4 can be further simplified in the case of emerging diseases. More precisely:
 \begin{Corollary}
 	If, at the beginning of the epidemic, we have no recovered people and a small number of infections, that is, $x(0) + y(0) = 1, x(0) \sim 1$, then condition \eqref{herdimmunitycondition} becomes:
 	\begin{equation}
 	\label{alternativeherdcondition}
 	 c\int_0^{1} \frac{\mathrm{d}u}{a(u)}> 1.
 	\end{equation}
 	\end{Corollary}
 
There are at least two important special cases under which condition \eqref{alternativeherdcondition} holds automatically, which are discussed below. 
\begin{rem}
If $a(1)<c$, then because of the increasingness of $a(\cdot)$, we have:
$$
c \int_0^{1} \frac{\mathrm{d}u}{a(u)}>  \frac{c}{a(1)} >1.
$$
Thus, if at the beginning of the epidemic the linearized version of model \eqref{linearized} is stable, then the herd immunity will necessarily be attained ultimately. 
 \end{rem}
 		
 \begin{rem}
 	If $\frac{1}{a(x)}$ is not integrable at $x=0$, \eqref{alternativeherdcondition} is also satisfied.
% If instead $a(1)>c$, then condition \eqref{alternativeherdcondition} may or may not be satisfied. Then another simple, sufficient but unnecessary condition for inequality \eqref{alternativeherdcondition}, or even \eqref{herdimmunitycondition} to hold is when $\frac{1}{a(x)}$ is not integrable at $x=0$, that is, the left hand side of \eqref{herdimmunitycondition} is equal to $+ \infty$. On the other hand, it is easily checked that condition \eqref{herdimmunitycondition} is satisfied for any value of $x(0)$ and $y(0)=1-x(0)$, if and only if $\frac{1}{a(x)}$ is not integrable is non integrable at zero. Note also that in this case, even though 0 remains a pole of the denominator of the integrand in equation (2.4), since another pole $v^*$ exists and is larger than 0, the state 0 will not be reached or approached by $x(t)$ when $t$ goes to infinity. 
 \end{rem}
\textbf{Example 1 (cont.)~:} If $a(x) = \alpha x$, the condition for herd immunity is automatically satisfied since $\frac{1}{x}$ is non integrable at zero. To look for the herd immunity level $x(\infty)$, we rewrite equation \eqref{theequation} into:
$$
\begin{array}{llcl}
 & x(\infty) & - & [x(0) + y(0)] = \Frac{c}{\alpha} \log [x(\infty)/x(0)] \\ \\
 \Longleftrightarrow & x(\infty) &=& x(0) \exp \{ - \Frac{\alpha}{c} [x(0) + y(0) - x(\infty)]\}.
\end{array}
$$

This equation depends on parameters $\alpha, c,$ through the ratio $R_0 = \alpha/c$, called the (initial) reproductive number (or reproduction number). Then the formula for $x(\infty)$ becomes~:

\begin{equation}
\label{herdimmunitybasicsir}
x(\infty) = \exp \{ - R_0 [1-x (\infty)]\},
\end{equation}
if $x(0) \approx 1$. 
This relation has been first derived in Kermack, McKendrick (1927) [see also Ma, Earn (2006)]. The fact that the ultimate size $x(\infty)$ depends only on the reproductive number $R_0$ means that the benchmark SIR model is very constrained. Indeed, $R_0$ has an interpretation close to $t=0$, by considering the stability condition of the linearized model [see equation \eqref{linearized}]. In other words, in the standard SIR model, the long run properties of the dynamic are characterized by its short run properties only. This is no longer the case for other transmission functions $a$.
\\
[1em]
%\begin{rem} The standard SIR model in Example 1 is very constrained. Indeed, the final size $x(\infty)$ depends on the reproductive number  $R_0$ only, which has an interpretation close to $t=0$, by considering the linearized version. In other words, in the standard SIR model, the long run properties of the dynamic are characterized by its short run properties only, which is a very restrictive consequence of this model. This is no longer the case for other choices of transmission function $a$.\end{rem}
%\begin{rem}
\textbf{Remark 3 (cont.)} \textit{A sufficient condition for $\frac{1}{a(x)}$ to be non integrable at $x=0$ is that $\frac{a(x)}{x}$ is nondecreasing in $x$, that is the case of mover-stayer phenomenon. Indeed, in this case we have:
$$
\frac{1}{a(x)} \geq \frac{1}{a(1)x}, \qquad \forall x \in ]0, 1[,
$$
which is therefore non integrable. }
\vspace{1em}
%\end{rem}

%\setcounter{rem}{4}

\textbf{Example 2 (cont.)~:} If $a(x) = \alpha x^\beta, \beta>1,$ we have~:

$$
x(\infty) - [x(0) + y(0)] = \Frac{c}{\alpha (1-\beta)} [x(\infty)^{1-\beta} - x(0)^{1-\beta}].
$$

At the beginning of the epidemic, if $x(0) \approx 1$, this equation becomes~:

\begin{equation}
\label{herdimmunitypowerfunction}
1-x(\infty)^{1-\beta} =  R_0(1-\beta) [1-x(\infty)].
\end{equation}
Thus the level of herd immunity $x(\infty)$, if it exists, depends on both parameters $R_0$ and $\beta$, not on $R_0$ alone. By Proposition 4 and Corollary 2, equation \eqref{herdimmunitypowerfunction} admits a unique solution over $]0, 1[$, if and only if:
\begin{equation}
\label{conditionherdimmunitypowersir}
\beta> 1- \frac{c}{a}=1-\frac{1}{R_0}.
\end{equation}
This condition involves both $R_0$ and $\beta$, and there are several special cases that deserve mentioning.
\begin{enumerate}[$i)$]
 \item If $\beta>1$, inequality \eqref{conditionherdimmunitypowersir} is satisfied for any $R_0>0$. This is the sufficient, but unnecessary condition, analyzed in Remark 3. 
\item If $R_0<1$, then again inequality \eqref{conditionherdimmunitypowersir} holds for any $\beta>0$. This is a consequence of Remark 2 above. 
\item If $R_0 \geq 1$, then the admissible domain of $\beta$ to ensure herd immunity is $]1-\frac{1}{R_0}, \infty[$.
\end{enumerate}
\vspace{0.5em}

\begin{rem}
	 In Example 2 above, we have retained the definition $R_0=a/c$. It should be emphasized that this is for expository purpose only, and that this ratio in this model no longer has the interpretation of reproduction number, that is the expected number of secondary infections due to an individual infected at a given date $t$. This differs from  the baseline SIR model of Example 1, in which $\alpha$ is the intensity of generating secondary infection and the duration of the infectious period is exponential with parameter $c$ and hence $R_0=a/c$ is the reproduction number. In Example 2, on the other hand, since the intensity of generating secondary infection is no longer constant, the reproduction number is also time varying (see Section 2.5). However, when $x(0)$ is close to 1, $R_0=a/c$ can still be viewed as the \textit{initial} reproduction number at $t=0$. 
\end{rem}

%\textbf{Yang: is it better to move the definition $R_0=a/c$ immediate after the linearized ODE, at the end of section 2.1? }

\begin{rem}It is possible to derive directly equation \eqref{theequation} from equation \eqref{fromytox} in Corollary 1, without looking for the closed form of $x(t)$. Indeed, by applying \eqref{fromytox} for $t=\infty$ and taking into account the fact that $y(\infty) = 0$, we get:
	$$
	-x (\infty) + 1 + c \Int^{x(\infty)}_{1} \Frac{du}{a(u)} = 0,
	$$
	that is equation \eqref{theequation}. This is the standard approach used to derive the final size in practice [see e.g. Anderson, May (1991), Ma, Earn (2006), Section 2, Miller (2012), p2126]. This approach provides the necessary form of the final size, but does not prove the convergence of $x(t)$ to this value, when $t$ tends to infinity. This convergence is a consequence of Proposition 4 and Corollaries 1, 2 under herd immunity. It can also be derived indirectly by using Lyapunov functions [Korobeinikov, Wake (2002), O'Regan et al. (2010)].
\end{rem} 

\subsection{Stability of final size $x(\infty)$ with respect to $x(0)$}
Throughout this subsection we assume that $x(0)+y(0)=1$, and condition \eqref{alternativeherdcondition} is satisfied, so that the existence of herd immunity is guaranteed for $x(0)$ sufficiently close to 1. In this section we study the stability of the ultimate size $x(\infty)$ with respect to the initial state of the population-at-risk, characterized by $x(0)$.   %and if $x(0)<1$, then $x(\infty)>0$. %On the other hand, in the limiting case $x(0)=1$, equation \eqref{herdimmunitycondition} has another trivial solution that is $v=1$. 
We use the notation $x(0)=1-\epsilon$, where $\epsilon$ is typically small, and denote the final size by $x_\epsilon(\infty)$ to emphasize its dependence on $\epsilon$. Since the linearization analysis [see equation \eqref{linearized}] shows that the dynamics of the system is quite different depending on the sign of $\frac{a(1)}{c}-1$, we expect also two different regimes for $x_\epsilon(\infty)$ when $\epsilon$ goes to zero. %Specifically, we have the following proposition: 

%Therefore, one natural question is: {what is the behavior of $x_\epsilon(\infty)$, when $\epsilon:=1-x(0)$ goes to zero?\footnote{Here we use $x_\epsilon(\infty)$ to emphasize the dependence of $x(\infty)$} on $\epsilon$.} We have the following result:

\begin{prop}
\label{stabilityprop}
	\begin{itemize}
		\item If $a(1)>c$, $x_\epsilon(\infty)$ converges to a constant $x_0(\infty) \in ]0, 1[$, when $\epsilon$ goes to zero. Moreover, $x_0(\infty)$ is the unique root over $]0, 1[$ of equation:
		$$v-1-c \int_{1}^v \frac{\mathrm{d}u}{a(u)}=0,$$
		and satisfies $x_0(\infty)< a^{-1}(c)<1.$ 
		\item If $a(1) \leq c$, then $x_\epsilon(\infty)$ converges to 1, when $\epsilon$ goes to zero. %Roughly speaking, this corresponds to the case where the epidemic extincts very quickly, with $x(\infty)$ very close to 1 if $\epsilon$ is sufficiently small.
	\end{itemize}
\end{prop}
\begin{proof}
	See Appendix B.2.
\end{proof}

This result is an enhanced version of Proposition 2. Roughly speaking, if $a(1)>c$, the initial state is unstable, and the epidemic explodes at the beginning with $y(t)$ increasing for small $t$, and the final herd immunity corresponds to the state where ``quite a lot of" individuals are infected and then immunized. Proposition 5 above says that this final state $x_\epsilon(\infty)$ is a continuous function of $\epsilon$. If, on the other hand $a(1) \leq c$, then the initial state is stable [or attractive, see Boatto et al. (2018)]. 
\vspace{1em}

\textbf{Example 1 (cont.) :} If $a(x) = \alpha x$, then $\frac{a(1)}{c}=\frac{\alpha}{c}=R_0$, which determines the stability (or the lack thereof) of the initial state, $x(0) \approx 1, y(0)=1-x(0)$. This result is well known for the basic SIR model [see e.g. Boatto et al. (2018), Thm 1, and references therein].
\\
%the condition for herd immunity at the beginning of the epidemic is~:$R_0 = \Frac{\alpha}{c} > 1,$
%with $x(\infty)=\frac{1}{R_0}$ [see Hens et al.(2020) eq (3.1)].

%This condition can be compared with the condition of explosion at the beginning of the epidemics and understood from the change of time used in deriving the closed form expression of the solution (see appendix). The calendar time has no epidemiological interpretation and has to be replaced by the epidemiological intrinsic time $x(t)$ that accounts for the speed of transmission. %When $R_0$ is large, the speed of transmission is higher, making the endemic equilibrium value reached earlier in intrinsic time, then without all individuals being touched.

\textbf{Example 2 (cont.)~:} If $a(x) = \alpha x^\beta$, we have again  $\frac{a(1)}{c}=\frac{\alpha}{c}=R_0$. 
Then we can refine the discussion below equation \eqref{conditionherdimmunitypowersir} as follows:
\begin{itemize}
\item If $R_0 \leq 1$, the final size $x_\epsilon(\infty)$ is close to 1 for small $\epsilon$.
\item Otherwise, if $R_0>1$, then depending on the value of $\beta$ we have three cases:
\begin{itemize}
		\item if $\beta \geq 1$, the hazard function $\frac{a[x(t)]}{x(t)}=[x(t)]^{\beta-1}$ is decreasing in time. Then there is herd immunity, with $x_\epsilon(\infty)$ away from 1. 
	\item if $\beta \in ]1- \frac{1}{R_0}, 1[$, the hazard function is increasing in time, but at a mild rate in $x(t)$. Then there is herd immunity, with $x_\epsilon(\infty)$ away from 1. 
		\item if $\beta \leq 1-\frac{1}{R_0}$, the rate at which the hazard function increases in $x$ is too large, the herd immunity is not attained and we have total infection $x_\epsilon(\infty)=0$ in finite time. 
\end{itemize} 
In other words, if $R_0>1$, there is an abrupt transition from the case $x_\epsilon(\infty)\approx 1$ to the case $x_\epsilon(\infty)=0$ around the point $\beta=1-\frac{1}{R_0}$. %In this latter limiting case, the condition of herd immunity is not satisfied
\end{itemize}
In Proposition \ref{stabilityprop}, with $x(0)=1-\epsilon$ and $c$ fixed, there may exist a discontinuity when $c=a(1)$. Let us now consider the small $\epsilon$ behavior when $x(0)=1-\epsilon$ and $c_\epsilon\rightarrow a(1)^{+}$, or $a(1)^{-}$. 
%Thus for this model, the fact that $R_0=1$ is a limiting case in Proposition 4 echoes the three different cases discussed for equation \eqref{conditionherdimmunitypowersir}.  
%the condition for herd immunity is the same : $R_0 = \Frac{\alpha}{c} > 1$. 
%Finally, note that in the above Proposition, the limiting case $a(1)=c$ corresponds to a herd immunity that is close to zero, therefore, by mimicking the same proof, it is straightforward to deduce that 

%	Let us now assume function $a(\cdot)$ given, and  let $\epsilon$ go to zero, and $c$ go to $a(1)$ simultaneously. We have two cases
\begin{Corollary}
	For a given function $a(\cdot)$, if
	\begin{itemize}
		\item either initially $c_\epsilon<a(1)$ and $c_\epsilon$ increases towards $a(1)^{-}$ when $\epsilon$ goes to zero, 
				\item or initially $c_\epsilon>a(1)$ and $c_\epsilon$ decreases towards $a(1)^{+}$ when $\epsilon$ goes to zero, 
	\end{itemize}
then $x_{\epsilon, c}(\infty)$ goes towards $1$. 
\end{Corollary}
This result means that even if at the beginning the epidemic is explosive but mildly,  the epidemic can still be well controlled (in the sense that the total infected proportion $1-x(\infty)$ is small), so long as the initial infections proportion $y(0)=\epsilon$ is 
 small enough. %If, instead the epidemic is non explosive at the beginning, then ev

\subsection{The reproduction number(s)}
There are two definitions of reproductive number in the epidemiological literature: the basic one and the effective one. These notions are heavily used in empirical studies as measures to track and control the evolution of various epidemics over time. However, their use is not without criticism [see e.g. Li et al. (2011)]. In this subsection we first give the formal definition of these numbers and compute them for the SPSIR model. Then we show that the aforementioned Examples 1 and 2 are particularly instructive to discuss the limits of relying solely on these measures for disease tracking, especially after the initial outbreak period. %In particular, we will see that this analysis is particularly instructive, if we focus on the aforementioned Example 2.
\vspace{1em}

\textbf{Definition 1.} (see e.g. Farrington and Whithaker, 2003) The effective reproductive number, denoted $R_e(t)$, at any time $t$ of the epidemic, is the expected number of secondary infections caused by one typical new infection.
\vspace{1em}

Following this definition, we have:
\begin{equation}
\label{definition}
R_e(t)=\int_0^\infty a[x(t+u)] \exp(-cu) \mathrm{d}u,
\end{equation}
where $\exp(-cu)$ is the probability that an infected individual stay infectious after a duration $u$, and $a[x(t+u)]$ is the instantaneous number of infections generated by each new infectious individual. Indeed, with a constant rate $c$ of recovery, the infection duration has an exponential distribution with parameter $c$ and survival function $\exp(-cu)$.

If $c$ is large such that with large probability, the duration of the infectious period of an infected individual is small, then we have:
\begin{equation}
\label{linearised}
R_e(t)\approx \int_0^\infty a[x(t)] \exp(-cu) \mathrm{d}u=\frac{a[x(t)]}{c}.
\end{equation}
We immediately note that by definition, $\frac{\mathrm{d}y(t)}{\mathrm{d}t}>0$ if and only if $\frac{a[x(t)]}{c}>1$. This property, which is often interpreted as the analog of Proposition 2 for positive $t$, is the key motivation in the literature of defining $R_e(t)$. For instance, as Farrington and Whithaker (2003) put it, ``if $R_e(t) \leq 1$ , then, while infections still occur, for example by limited spread from imported cases, they cannot result in large epidemics". It is shown later on that this analog is not appropriate and should be used carefully.  

\begin{rem}
	The approximation used in equation \eqref{linearised} is based on the assumption that $c$ is large. This assumption is not necessarily realistic, since in practice, when most epidemics start, we have $R_e(0)=\frac{a(1)}{c}>1$, which means $c$ is small compared to $a(1)$.
	
	However, if function $a(\cdot)$ is differentiable, then it is possible to conduct a Taylor's expansion of equation \eqref{definition}. We have:
	\begin{align}
	R_e(t)&\approx  \frac{a[x(t)]}{c}    +\frac{a'[x(t)]x'(t)}{c^2} \nonumber \\
	&= \frac{a[x(t)]}{c} \Big[ 1-\frac{a'[x(t)]y(t)}{c}\Big].
	\end{align}
\end{rem}
The second term in the expansion depends not only on the current size of the susceptible population $x(t)$, but also on the current size of the infective population $y(t)$. At $t=0$, with $y(0) \approx 0$, the second time becomes negligible. %Note also that if $a(\cdot)$ is not differentiable, then intuitively the approximation in \eqref{linearised} is very poor. This is for instance the case if $a(\cdot)$ is replaced by a stochastic process $A(\cdot)$ that is nowhere differentiable. (pas a la bonne place)
\vspace{1em}

\textbf{Definition 2.} The basic reproduction number, denoted $R_0(t)$, is defined as:
\begin{equation}
	\label{effective}
R_0(t)= \frac{R_e(t)}{x(t)}.
\end{equation}
Hence, under approximation \eqref{linearised}, we have $R_0(t)=\frac{1}{c}\frac{a[x(t)]}{x(t)}$. Thus we recover the ratio $\frac{a[x(t)]}{x(t)}$, which is the instantaneous hazard function discussed in Section 2.1. In particular, if at $t=0$ with $x(0) \approx 1$, we have: $R_0(0)=R_e(0)$, which was denoted $R_0$ in Examples 1 and 2 above.  

 $R_0(t)$ is obtained from $R_e(t)$ by accounting for the diminishing susceptible population over time. The choice of the normalization term $x(t)$ can be motivated by the fact that the instantaneous number of infections $a[x(t)]$ is decomposed into the product of $\frac{a[x(t)]}{x(t)}$, that is the \textit{instantaneous number of contacts (including with those that are no longer susceptible)}, times $x(t)$, which is the probability that a contact is still susceptible (and hence will be infected almost surely upon contact).\footnote{If the size of the country is fixed, there is automatically a diminution of the density of individuals at risk. In some sense, there is an automatic social distancing between the susceptible and infectious individuals.} Hence $\frac{a[x(t)]}{x(t)}$ is a measure that is adjusted for the current size $x(t)$ of the population-at-risk.
%$R_0(t)$ is the expected proportion of susceptible people that will be infected by an infectious individual, and is therefore a measure that is free of the current size $x(t)$. 
In particular, in the baseline SIR model, the instantaneous number of contacts of an infectious individual $\frac{a[x(t)]}{x(t)}= \alpha$ is constant. %The motivation of using $x(t)$ as the normalization term is the following: in this baseline SIR model, the average number of contacts over the infectious period of a patient is $\alpha/c$, but at time $t$, among all the contacts, only a fraction $x(t)$ are susceptible. 
\vspace{1em}

Thus, depending on the problem, one of the two reproductive numbers might be more suitable. If we want to predict the number of infectious individuals, which is an important indicator of hospitalization or ICU needs, say, then $R_e(t)$ is more suitable; if instead one would like to measure the infectivity of the infectious population, then we have to account for the fact that only a fraction (in fact $x(t)$) of their contacts can lead to infections and hence $R(t)$ is more adapted.
\vspace{1em}

Next, let us look at an example in which both reproductive numbers should be interpreted with care, especially when the epidemic has already reached an advanced stage.  

%The $R_0(t)$ and $R_e(t)$ are easier to interpret around zero (due to the linearized ODE system). For positive time $t$, we have seen that $R_0(t)$ is constant so is also easy to interpret for the baseline SIR model. 
%Having provided the formal definitions of the effectivea and basic reproduction number, let us now look at their properties at a positive time $t$, for a non-baseline SPSIR model. For expository purpose, let us consider Example 2, for instance, with $a(x)=\alpha x^\beta$. 
\vspace{1em}

\textbf{Example 2 (cont.)} At time $t>0$, we have, under approximation \eqref{linearised}:
\begin{enumerate}[$i)$]
	\item $R_e(t)=\frac{\alpha}{c} [x(t)]^{\beta}$. This quantity goes to zero if $x(t)$ goes to zero, so long as $\beta>0$. In Section 2.3 we have seen that, if $\beta<1-1/R_0(0)$, then the final size goes to zero in finite time. But, even if at a time $t$ the current $x(t)$ is such that $R_e(t)<1$, it does not guarantee a  ``slow down" of the epidemic progression, nor the existence of an eventual herd immunity. %does not necessarily guarantee the existence of herd immunity outside zero. 
	Intuitively, even though $y(t)$ decreases for $x(t)$ sufficiently small, the measure $R_e(t)$ fails to account for the equally fast decrease rate of $x(t)$ when $x(t)$ is near 0.  This is very different from what arises at the initial time $t=0$, since we have shown that for any SPSIR model with arbitrary function $a(\cdot)$, there is herd immunity, with a final size $x(\infty)>0$ whenever $R_0(0)<1$. 
	\item The interpretation of $R_0(t)=\frac{\alpha}{c} [x(t)]^{\beta-1}$ is not straightforward, neither. Depending on $\beta$, we have two cases:
	\begin{itemize}
		\item If $\beta>1$, there is herd immunity, and $R_0(t)$ goes to zero if $x(t)$ goes to zero. This is the ``well-behaved" case.
		\item If $\beta$ is between $1-1/R_0(0)$ and $1$, then $R_0(t)$ goes to infinity if $x(t)$ can go to zero. However, we know that in this case there is herd immunity with a final size $x(\infty)>0$. Therefore, the statistic $R_0(t)$ alone does not allow us to determine the explosiveness/convergence of the system. 
	\end{itemize}
\end{enumerate}
 
	The above ``paradox", that is the failure of $R_e(t)$ and $R_0(t)$ to characterize the explosiveness/convergence of the system for positive $t$ echoes similar criticism on $R_0=R_0(0)$ in the literature [see e.g. Tildesley and Keeling (2009), Li et al. (2011)], although they concern different epidemiological models. In particular, the failure is not due to the approximation error in \eqref{linearised} alone, but also due to the Definition 1 of these numbers, or equivalently to equation \eqref{definition}. Intuitively, what counts the most is not the average number of secondary infections \textit{per} infectious individual, but rather the average number of new infections due to \textit{all} existing infectious individuals [which we denote by $R_a(t)$], as well as its normalized version $R_a(t)/x(t)$ to account for the diminishing size of the susceptible population. %This difference is typically not an issue 
	At the beginning of the epidemic, when $y(t)$ is close to zero, the probability that   a given susceptible individual has contact with more than two infectious individuals is negligible and thus each new infection can be attributed to one existing infectious case. Then $R_a(t)$ is roughly equal to $R_e(t) y(t)$ (up to normalization by the total population size $N$). Then $R_e(t)$ is a convenient, scale-free measure of the short-term evolution of $y(t)$, and hence of $R_a(t)$. When $t$ moves away from 0 and $y(t)$ becomes large and/or $x(t)$ becomes small, contact tracing becomes more and more difficult since each susceptible individual can have contact with many infectious cases and it is no longer possible to attribute a new infection to one specific infectious case only. In other words, $R_a(t)$ is no longer close to $y(t)R_e(t)$. Hence, even   $R_e(t)$ no longer characterises the future dynamics of $R_a(t)$ well. %A second possible explanation is through the lens of the information set. For positive $t$, the future evolution of the system is completely characterized by the pair $(x(t), y(t))$, but not by $R_e(t)$ alone, which only depends on $x(t)$. Hence focusing on these reproduction numbers alone is inappropriate since it leads to significant information loss. %Finally, when $x(t)$ is significantly smaller than 1, even $R_a(t)$ becomes inappropriate since it fails to account for the diminishing size and normalizing it by $x(t)$, as is done in Definition 2, seems more appropriate. 
	
%	Therefore, because of the large number of competing risks, $R_e(t)$ becomes a rather poor measure, especially when $x(t)$ is small compared to $y(t)$. 
 
 Finally, note that these notions differ significantly from the so-called instantaneous reproductive number introduced in Fraser (2007) to get a simple computation of a reproductive number. This difficult notion has been popularized by Cori et al. (2013), but is misleading, since it does not account for the nonlinear features existing in a SIR model [see the discussion in Eliott et al. (2020)].

\section{Stochastic Transmission}
\setcounter{equation}{0}\def\theequation{3.\arabic{equation}}
This section extends the model of Section 2 by allowing the transmission function $a(\cdot)$ to be stochastic. 
\subsection{A stylized fact}

When the standard SIR model is used, and $x(0) \sim 1$, with $y(0) \sim 0,$ equation (2.5) reduces to~:

$$
y(t) = 1-x(t) + \Frac{c}{\alpha} \log x(t),
$$

\noindent or equivalently~:

\begin{equation}
  1-x(t) - y(t) = z(t) = -(1/R_0) \log x(t).
\end{equation}

It is usual in practice to deduce from observed frequencies $\hat{x} (t), \hat{z}(t),$ estimators of the reproductive number day per day as for instance~:

\begin{equation}
\label{OLS}
  \hat{R}_{0,t} = -\log \hat{x}(t)/\hat{z}(t),
\end{equation}
or more generally by applying ordinary least square to the approximated regression model $\hat{z}(t)=-\frac{1}{R_0}\log \hat{x}(t) $ in a rolling way.\footnote{Formula \eqref{OLS} corresponds to this practice with a rolling window of one day.} Then, it is generally observed that these estimated values are very erratic over time. This "volatility" can be explained in different ways as follows~:
\vspace{1em}

i) There is an effect due to the replacement of the theoretical probabilities by the observed frequencies, even if the cross-sectional dimension, that is the size of the population, is large. Indeed, at least at the beginning of the epidemic, $z(t),$ and also $\hat{z}(t),$ are small, and this effect is amplified on $\hat{R}_{0t}$ by the nonlinear transformation (3.2). {This is called the \textit{demographic stochasticity} in the literature [see e.g. Breto et al. (2009)].} \vspace{1em}

ii) An alternative explanation is that the parameters in the SIR model are time dependent, especially due to endogenous health policies of lockdown, social distancing, or management of the number of intensive care unit beds. \vspace{1em}

iii) The transmission function can also be misspecified. For instance, under the SPSIR model, we get from (2.5)~:
$$
z(t) = -c \Int^{x(t)}_1  \Frac{du}{a(u)},
$$
instead of $z(t) = -(1/R_0) \log x(t).$ Therefore, $\hat{R}_{0t}$ approximates $\tilde{R}_{0t} = \log x(t) / c \Int^{x(t)}_1 \Frac{du}{a(u)}$, instead of $R_0$. This quantity depends on time $t$, due to the nonlinear transmission function $a$. Intuitively, there is a one-to-one relationship between $\tilde{R}_{0t}, 0 < t<T$, and function $a$ for $x$ between 1 and $x(T)$. The series of $\tilde{R}_{0t}, 0 < t<T$, can be seen as a functional summary statistics providing consistent information on a segment of the transmission function. But this quantity differs from the true reproductive numbers $R_0(t)$ or $R_e(t)$ introduced in Section 2.5 [see eq. (2.14)]. \vspace{1em}

iv) This observed "volatility" can also indicate that the transmission function itself is uncertain, i.e. stochastic. {Thus we have the so-called \textit{environmental stochasticity} on top of the demographic (i.e. observational) stochasticity. It is widely accepted in the epidemiological literature that the demographic stochasticity alone is often insufficient to explain the deviations of the data from their predictions given by deterministic models [see e.g. Breto et al. (2009)]. }This feature is introduced below by considering SIR model with stochastic transmission.

\subsection{SIR model with stochastic transmission}

A SIR model with stochastic transmission is defined by the stochastic dynamic system~:

\begin{equation}
  \left\{
  \begin{array}{lcl}
  \Frac{d X(t)}{dt} & = & - A [X(t);\omega] Y(t), \\ \\
  \Frac{dY(t)}{dt} & = & A [X(t); \omega] Y(t) - c Y(t), \\ \\
  \Frac{dZ(t)}{dt} & = & c Y (t),
  \end{array}
  \right.
\end{equation}
\noindent where the transmission function $a(.) = A(.;\omega)$ depends on a state of nature $\omega$. Then the solutions $X(t), Y(t),Z(t)$ become also random. The stochastic transmission function $A(x;\omega) \equiv A(x), x \in [0,1],$ is a stochastic process indexed by the theoretical proportion of susceptible.\vspace{1em}

\begin{Ass} 
	\label{stochastictransmission}
	$i)$ $A(0) = 0$.
	 $ii)$ Process $A(\cdot)$ is non degenerate, that is non deterministic, and almost surely (a.s.) strictly increasing on $[0,1]$.
\end{Ass}

Under Assumption \ref{stochastictransmission}, the results in Section 2 can be applied for (a.s.) any state of nature $\omega$. For instance (2.2) can be applied with stochastic $X$ and $A$ instead of scalar small $x$ and $a$.

Similarly the condition in Corollary 3 on the comparison to 1 of a stochastic latent variable $A(1)/c$ implies now an uncertainty on the existence of herd immunity. Moreover, when it exists, the level of herd immunity itself is random. In other words, the final size $X(\infty)$ is random with a distribution, which is a mixture of a point mass at zero, and a continuous distribution on $(0,1)$. The size and time of the peak in the proportion of infected people are stochastic as well.\vspace{1em}

The introduction of uncertainty through a stochastic transmission is an alternative to other stochastic SIR models considered in the literature. For instance, the differential system (2.1) can be transformed into a stochastic differential system (SDS) as~:

\begin{equation}
  \left\{
  \begin{array}{lcl}
  \Frac{dx(t)}{dt} & = & - \alpha x(t) y(t) - \sigma x(t) y(t) dW(t), \\ \\
  \Frac{dy(t)}{dt} & = & [\alpha x(t)-c] y(t) + \sigma x(t) y(t) dW(t), \\ \\
  \Frac{dz(t)}{dt} & = & c y (t),
  \end{array}
  \right.
\end{equation}

\noindent where $\alpha, c, \sigma$ are positive scalars and $W$ a Brownian motion [see e.g. Jiang et al. (2011), Rifhat et al. (2017), Xu, Li (2018), El Koufi et al. (2019)].

This automatic SDS extension has at least three drawbacks:
\begin{enumerate}[$i)$]
	\item It is difficult to derive the conditions ensuring solutions between $0$ and $1$, as well as the existence conditions of herd immunity [see e.g. El Koufi et al. (2019), Sections 2.3].

\item More important, this extension does not focus on the "functional parameter of interest", that is the transmission function. In other words such an extension is not structural enough.

\item In particular, it is not appropriate for an analysis of systemic factor ($A$ in our framework) and for comparing the health policies to control this factor.
\end{enumerate}

Another literature [see e.g. Dureau et al. (2013)] extends the benchmark SIR model by assuming that coefficient $\alpha$ is random and time-varying, that is, for $x(t)$: 
\begin{equation}
\label{randomcoefficient}
\frac{\mathrm{d} x(t)}{\mathrm{d}t}=-\alpha(t) x(t) y(t),
\end{equation}
where the dynamics of the stochastic process $(\alpha(t))$ is exogenous. As a comparison, if we equivalently rewrite the first equation of our model as $\frac{\mathrm{d} X(t)}{\mathrm{d}t}=-\frac{A[X(t);\omega]}{X(t)} X(t) Y(t)$, we see that the ratio $\frac{A[X(t);\omega]}{X(t)}$, which plays the same role as $\alpha(t)$ in model \eqref{randomcoefficient}, is also time-varying, but \textit{endogenous}. 
%\textbf{Thus, contrary to some belief [see e.g. Novozhilov (2008), Proposition 1], the exogenous and endogenous modelling are different, at least in a stochastic framework. (not very clear)}
 Depending on the application, this endogeneity might be more desirable. For instance, in the case of COVID-19, the potential change of public health policies are largely influenced by the current state of the epidemic. Moreover, models of type \eqref{randomcoefficient} are usually less tractable than the SPSIR model that we propose. 
 %\textbf{For instance, Boatto et al. (2018) show that if process $(\alpha(t))$ is periodic, the solution of \eqref{randomcoefficient} might have chaotic behavior.} (not really true. The model of Boatto et al. is not SIR but SEIR)
\vspace{1em}
\begin{rem}
A way to reconcile these two approaches is to replace $A[x(t), \omega]$ by $A[x(t),t, \omega]$, where $A[\cdot,\cdot, \omega]$ is a bivariate random field that depends stochastically on both of its first two arguments. This extension is out of the scope of the present paper. 
\end{rem}

\subsection{The dynamics of stochastic transmission}

The SPSIR model (3.3) is completed by modeling the dynamic of stochastic transmission. In order to satisfy Assumption 1, these dynamics are constructed from processes with positive increments. Extensions of the SIR model with constant transmission rate are obtained by considering stationary increments. Then a state deformation is applied to extend to nonconstant transmission rate.\vspace{1em}

\textbf{Example 4~: Gamma process.}\vspace{1em}

Process $A(x), x \in [0,1]$ can be assumed a gamma process, also called increasing Levy process [Revuz, Yor (2001)]. This process is largely used to perform stochastic time change, as in the variance-gamma model [Madan, Carr, Chang (1998)] [see the time change used in the proof of Proposition 3 given in Appendix A  and the role of time changes in epidemiological models Bohner, Peterson (2003)]. It is also used for modelling the risk-neutral dynamics, i.e. the price of time [Clement, Gourieroux, Monfort (2000)].

This is a stochastic process with independent increments, such that the distribution of $A(x+dx)-A(x)$ is a gamma distribution $\gamma (\theta dx, \lambda)$, with degree of freedom $\theta dx, \theta >0,$ and scale parameter $\lambda, \lambda >0$. The expectation of $A(x)$ is~:

$$
E[A(x)] = \Frac{\theta}{\lambda} x.
$$

Therefore in average we get the basic SIR model, with $\alpha = \theta/\lambda$. Compared to the basic SIR model, the SIR model with gamma stochastic transmission involves an additional parameter, that allows for capturing the uncertainty on transmission. The basic SIR model corresponds to the limiting case where $\theta \rightarrow \infty, \lambda \rightarrow \infty,$ such that $\theta/\lambda \rightarrow \alpha >0.$\vspace{1em}

\textbf{Example 5~: Integrated Cox-Ingersoll-Ross (ICIR) process}\vspace{1em}

An alternative modelling assumes that $A(x) = \Int^x_0 B(u) du$, where process $B$ satisfies a stochastic diffusion equation~:

$$
dB(x) = \mu [B(x)] dx + \sigma [B(x)] dW(x),
$$

\noindent where $W$ is a Brownian motion on $[0,1]$ and $\mu, \sigma$ are drift and volatility functions, respectively. To satisfy Assumption 1, the drift and volatility functions have to be chosen in order for process $Z$ to be positive. This arises for instance when $B$ is a Cox-Ingersoll-Ross (CIR) process [Cox, Ingersoll, Ross (1985)]~:

$$
dB(x) = - \gamma (B(x)-\alpha) dx + \sigma \sqrt{B(x)} dW(x),
$$

\noindent where $\alpha>0, \gamma >0, \sigma >0$, with $2\gamma \alpha > \sigma^2$.

We get $E B(x) = \alpha$ and $EA(x) = \alpha x.$ As in Example 4 we get in average the basic SIR model. The SIR model with ICIR stochastic transmission involves two additional parameters $\gamma$ and $\sigma$. Parameter $\sigma$ captures the transmission uncertainty and parameter $\gamma$ the transmission persistence. The basic SIR model is obtained in the limiting case~: $\gamma \rightarrow 0, \sigma \rightarrow 0$.\vspace{1em}

Other diffusion models can be used for $B$ as the constant Elasticity of Variance (CEV) model, where the volatility functions is $\sigma B(x)^\delta$, that is not constrained to be square root [Cox (1996)].
\vspace{1em}

\textbf{Example 6~: Time-changed stochastic processes} %

The models considered in Examples 4 and 5 can be extended by introducing another time change, following the approach of Lin et al. (2018). This introduces extra flexibility for, in particular, the final outcome (herd immunity or complete infection). For expository purpose, let us consider the extension of L\'evy subordinators that is a nonnegative, nondecreasing L\'evy process going to infinity when $t$ goes to infinity, such as gamma process. We denote by $(B_t)$ a L\'evy subordinator, with Laplace transform given by:
$$
\mathbb{E}[e^{v B_t}]=e^{t \ell(v)}, \qquad \forall t,
$$
where function $\ell(\cdot)$ is the Laplace exponent and is monotonous, and the admissible domain of $v$ includes at least all the negative real numbers. Then for any fixed $v$, we define a time- and scale-changed process $A_v$ through:
\begin{equation}
\label{definitionatheta}
A_v(x)=\alpha \exp\Big(v B_{K_v(x)}\Big), \qquad \forall x \in (0, 1)
\end{equation}
where $\alpha$ is a constant and $
K_v(x)= \frac{\log x}{\ell(v)}$ is the time change rule. It is easily checked that this stochastic process $(A_v(x))_{x \in (0 ,1)}$ is such that: $A_v(0)=0$, $A_v(1)=\alpha$. That is, the reproductive number at origin $t=0$ is deterministic and equal to $\frac{\alpha}{c}$. Moreover, for $x$ varying between 0 and 1:
 $$\mathbb{E}[A_v(x)]=\alpha e^{\ell(v) K_v(x) }=\alpha x. $$ 
 Therefore in average, we still get the basic SIR model.\vspace{1em}

Let us now study the existence of herd immunity. We have:
\begin{align}
	\int_0^1  \frac{\mathrm{d} x}{A_v(x)}&=\frac{1}{\alpha} \int_0^1 \exp(-v B_{\frac{\log x}{\ell(v)}}) \mathrm{d}x = \frac{\ell(v)}{\alpha} \int_{0}^\infty \exp(\ell(v)s-v B_s) \mathrm{d} s \label{exponentialoflevy}
\end{align}
where we have used the change of variable $ s=\frac{\log x}{\ell(v)}$, or $\mathrm{d}x= \ell(v)e^{\ell(v)s} \mathrm{d}s$. Thus we get the integral $ \int_{0}^\infty \exp(\ell(v)s-v B_s) \mathrm{d} s$, also called exponential functional of the L\'evy processes $(\ell(v)s-vB_s)_{s \geq 0}$. Its properties are well known in the probability community and, in particular, its distribution, and hence the probability of reaching herd immunity $\mathbb{P}[\int_0^1  \frac{\mathrm{d} x}{A_v(x)}> \frac{1}{c}]$ can be computed numerically [see e.g. Patie, Savov (2013) and references therein]. In some cases, we can also get a simple approximation of this latter probability. Indeed, since  $\int_0^1  \frac{\mathrm{d} x}{A_v(x)}$ is positive, we can apply Markov's inequality and get:
\begin{align*}
\mathbb{P}\Big[\int_0^1  \frac{\mathrm{d} x}{A_v(x)}> \frac{1}{c}\Big] &\leq c \mathbb{E}[ \int_0^1  \frac{\mathrm{d} x}{A_v(x)}] =  \frac{c}{\alpha}\int_0^1 \mathbb{E}[e^{-v B_{K_v(x)}}]\mathrm{d} x \\
&=\frac{c}{\alpha}\int_0^1 e^{\ell(-v)\frac{\log x}{\ell(v)}}\mathrm{d} x=\int_0^1 \frac{c}{\alpha}x^{\frac{\ell(-v)}{\ell(v)}}\mathrm{d} x=\frac{c}{\alpha} \frac{1}{\frac{\ell(-v)}{\ell(v)}+1}.
\end{align*}
provided that $\frac{\ell(-v)}{\ell(v)} >-1$. For instance, if $(B_x)_x$ is a gamma process, that is, if $\ell(v)=-\log(1-v), \forall v <1$, then we have:
$\frac{\ell(-v)}{\ell(v)}=-\frac{\log(1+v)}{\log(1-v)}$, which is larger than $-1$ if and only if $v$ is between 0 and 1. Thus, if $v$ further satisfies $\frac{c}{\alpha} \frac{\log(1-v)}{\log(1-v^2)}<1$, we get an upper bound of the probability of reaching herd immunity. This is an improvement compared to the baseline gamma process model discussed in Example 4, for which there is no result on the distribution of ${1}/{A(x)}$. %Nevertheless, these analyses are well beyond the scope of the paper [see e.g. Patie, Savov (2013) and references therein]. 

\vspace{1em}

The three examples above extend the basic SIR model with constant average transmission rate $\alpha$, that is, they all satisfy $\mathbb{E}[A(x)]=\alpha x$. These dynamics can be adjusted for extending the SPSIR model with nonlinear expected growth pattern for $\mathbb{E}[A(x)]$. This is easily done by applying a state deformation.\vspace{1em}

\textbf{Example 7~: State deformation}\vspace{1em}

Let us denote by $A^*$ a process with stationary positive increments as the gamma process of Example 4, or the integrated CIR process of Example 5. Then, the process defined by~:

$$
A(x) = A^* [a(x)],
$$

\noindent where $a$ is a deterministic increasing function, is such that~:

$$
E A(x) = a(x).
$$

%For instance, we extend the growth assumed in Example 2 by considering $A(x) = A^* (\alpha x^\beta).$\vspace{1em}

\section{Simulation}
This section illustrates by simulation the dynamic properties of a SIR model with stochastic transmission. 
\setcounter{equation}{0}\def\theequation{4.\arabic{equation}}
\subsection{Comparison of simulation methods}
Let us now explain how to simulate trajectories of a SIR model with stochastic transmission. As usual, we are mainly interested in daily values, that are $x(t)$, $y(t)$, $z(t)$, for $t=0,1,2,...$
They are obtained from the simulated values computed with a thin grid, with a timestep of $1/100$, say, keeping only the indexes multiples of $100$. We apply the following methods:

$i)$ We fix $x(0)=\epsilon$, $y(0)=1-x(0)$, $z(0)=0$, and simulate process $A$ from $1$ to $0$ on a thin grid of $1/100000$, say. Then we deduce $\{x(t), y(t), t=k/100\}$ from the integral formula (2.4), approximated by their Riemann sums on the grid for $A$, and keep only the values $[x(t),y(t)], t=0,1,...,T=400$, corresponding to slightly more than 1 year. We simulate the trajectory of $A(\cdot)$ as follows. $A(1-\epsilon)= \alpha=0.1$. Then the trajectory of $(A(x),x \in (0,1-\epsilon))$ is defined as a gamma bridge, that is a gamma process with degree of freedom parameter $\theta$, conditional on its terminal value $\alpha$.\footnote{The motivation of considering a gamma bridge, instead of a gamma process directly, is that at time $t=0$, the value of $A(x(0))=A(1-\epsilon)$ can be recovered first. } The simulated trajectory is provided below. 
\begin{figure}[H]
	\begin{center}
		\includegraphics[scale=0.3]{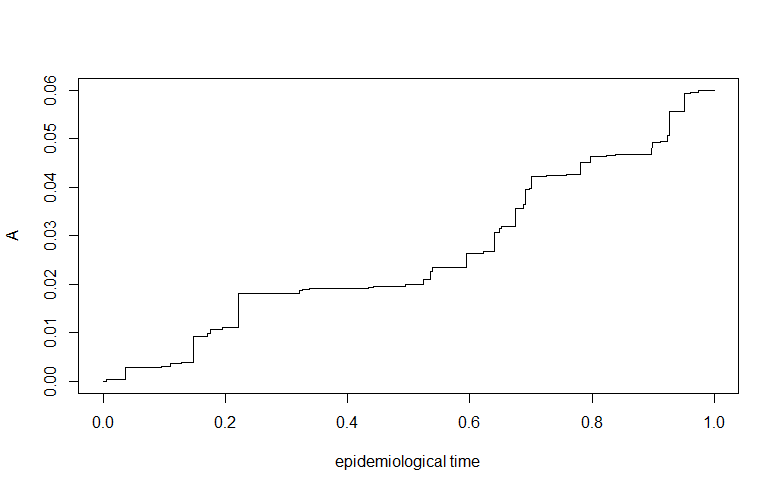}
		\caption{Trajectory of $A(\cdot)$. The parameters of the process are set as $\alpha=0.1, c=0.05,
			\theta=10,
			\epsilon=0.0001$.}
	\end{center}
\end{figure}

$ii)$ We also approximate the continuous time stochastic transmission SIR model by its Euler discretization (see Appendix C for the discussion of Euler discretization of the deterministic SIR model). Therefore, system (3.3) is replaced by the discrete time system of nonlinear recursive equations:
\begin{align} 
		X(n) - X(n-1)&=-h  A [X(n-1), \omega] Y(n-1),\\
	Y(n) - Y(n-1)&=h    A[X(n-1), \omega] Y(n-1) - cY(n-1) ,\\
		Z(n) - Z(n-1)&=h  c Y(n-1), 
\end{align}
where $h=0.1$. This approach avoids the resolution of integral equations at each date $t$. 

$iii)$ We also compute the values $[x(t),y(t)], t=0,1,...,T$ corresponding to the continuous time deterministic baseline SIR model with $EA(x)=\alpha x $ by means of the integral equation, where values of the parameters $\alpha$ and $c$ are the same as before.

$iv)$ Finally, we compute the deterministic analog of system (4.1), that is
\begin{align}
		x(n) - x(n-1)&=-h \alpha x(n-1)  y(n-1),\\
		y(n) - y(n-1)&=h \{ \alpha x(n-1)  y(n-1) - cy(n-1)\},\\
		z(n) - z(n-1)&=h c y(n-1), 
\end{align}
%with $\mathbb{E}[A(x)]=\alpha x$.
%\vspace{1em}
in which $A[X(n-1), \omega]$ is replaced by $\alpha x(n-1)$.

\paragraph{Stochastic SPSIR model.} Let us first first compare algorithms $i)$ and $ii)$ that are based on the same stochastic SPSIR model but different discretization strategies. The simulated trajectories $\big(X(t), Y(t)\big)$ are provided in Figure 2 below. 

\begin{figure}[H]
	\begin{center}
		\includegraphics[scale=0.4]{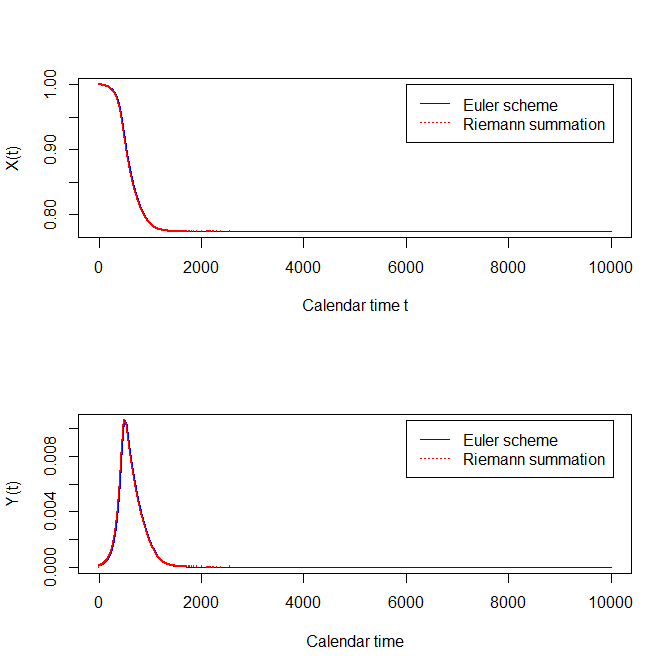}
		\caption{Trajectories of the processes under the Euler scheme and Riemann approximation, respectively. Parameter values are set to $N=100000,
			R_0=1.2,
			c=0.05.$ We also set the initial value to $\epsilon=0.0001$, with timestep $h=0.1$.}
	\end{center}
\end{figure}

Next, we plot the $y(t)$ in epidemiological time $x(t)$ for the two methods above. 

\begin{figure}[H]
	\begin{center}
		\begin{subfigure}{.6\textwidth}
		%	\centering
		\includegraphics[scale=0.25]{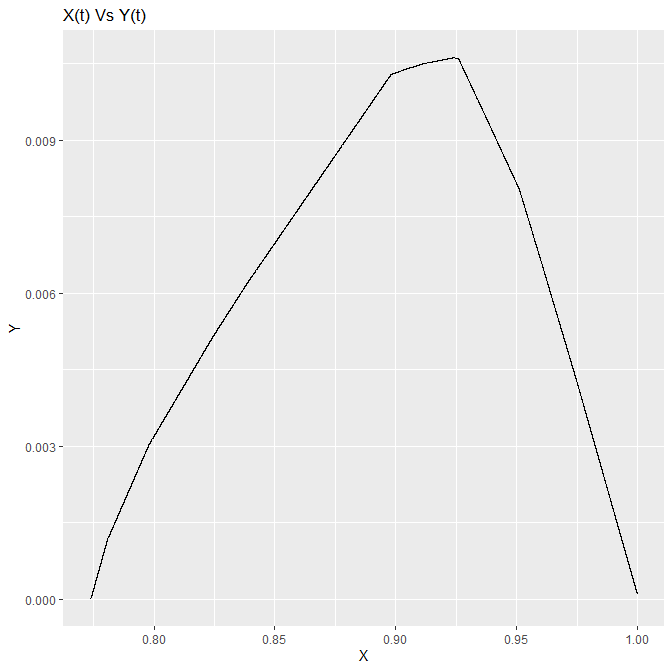}
	\end{subfigure}%
\begin{subfigure}{.6\textwidth}
		%		\centering
		\includegraphics[scale=0.25]{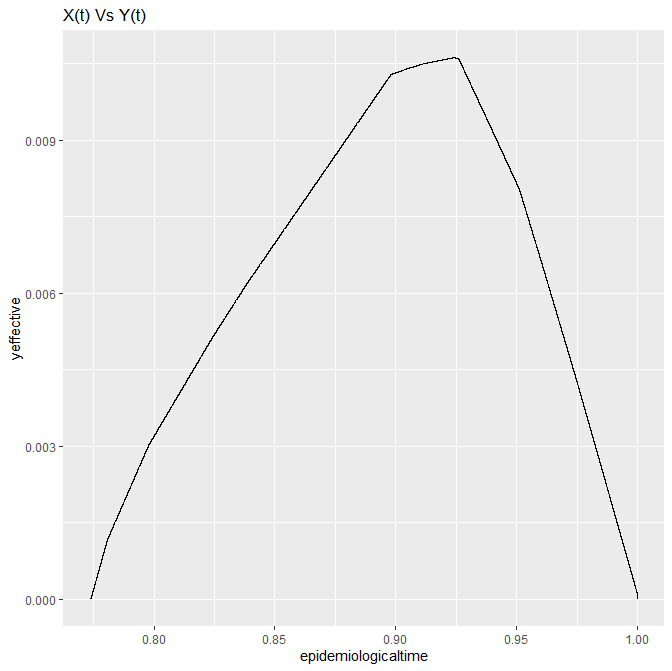}
			\end{subfigure}%
		\caption{Left panel: path of $(x(t), y(t))$ computed using Euler's scheme. Right panel: the same path computed using the integral equation. %Note, that to obtain the latter panel, only the integral equation involving $y(t)$ is used and there is no need to inverse the integral equation involving $x(t)$. 
		}
	\end{center}
\end{figure}

The two simulation methods provide close results. 

Let us now consider another experiment with different parameters. We change the value of $\epsilon$ to $0.0000008$, while holding other parameters constant. Figure 4 is the analog of Figure 2.

\begin{figure}[H]
	\begin{center}
		\includegraphics[scale=0.4]{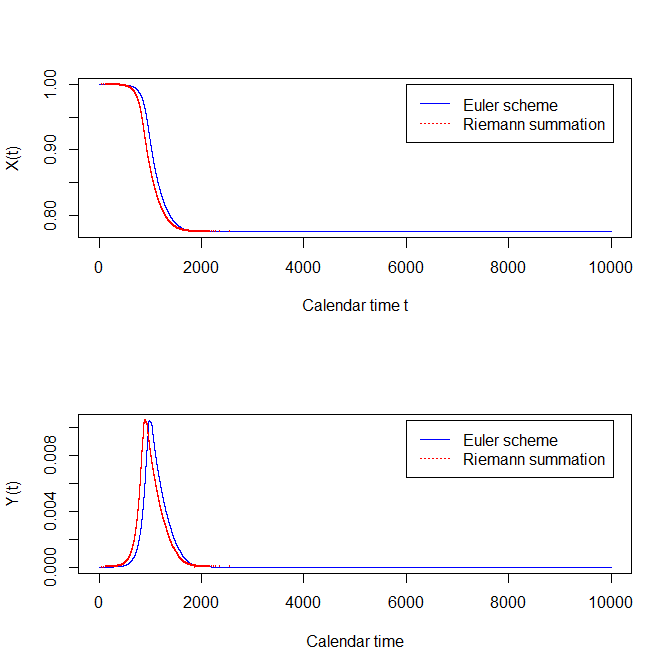}
		\caption{Trajectories of the processes under the Euler Scheme and Riemann approximation, respectively. Parameter values are the same as in Figure 2, with initial value set to $\epsilon=0.0000008.$}
	\end{center}
\end{figure}
We can check that the trajectory of the proportion of infectious individuals $y(t)$ for the Riemann summation is nearly identical  with that of Figure 2. However, the trajectory of $y(t)$ for the Euler scheme has shifted significantly. This shows the weakness of the Euler's discretization scheme. Indeed the time of the peak is defined by $A^{-1}(t)=c$, which should not significantly depend on $\epsilon$. %, since $A$ depends on $\epsilon$ only through a small time linear change and its trajectory is piecewise constant. 

To further support our claim of lack of robustness of the Euler's scheme, we report in Table 1 how this time of the peak changes as a function of $\epsilon$ for different values of $\epsilon$, by using the same simulated trajectory of $A(\cdot)$. 

\begin{table}[H]
	\begin{center}
		\begin{tabular}{|c|c|c|c|c|} 
			\hline
			$\epsilon$ & $8  * 10^{-6}$&	$8 * 10^{-5}$ &	 $8 * 10^{-4}$ & $8 * 10^{-3}$   \\ \hline
			Time of the peak & 1084 & 853 & 623 & 394 \\ \hline
		\end{tabular}
		\caption{Time of the peak for different values of $\epsilon$ under the Euler's scheme. We do not report the corresponding values for the Riemann approximation since the time of the peak ($\tau=1475$) does not change in $\epsilon$ for that method.}
	\end{center}
\end{table}

We have also checked whether increasing the time step in the Euler scheme would provide more accurate approximation, but the results are not modified if we change the value of $h$ from 0.1 to 0.01.

\begin{figure}[H]
	\begin{center}
			\begin{subfigure}{.6\textwidth}
				\centering	
		\includegraphics[scale=0.25]{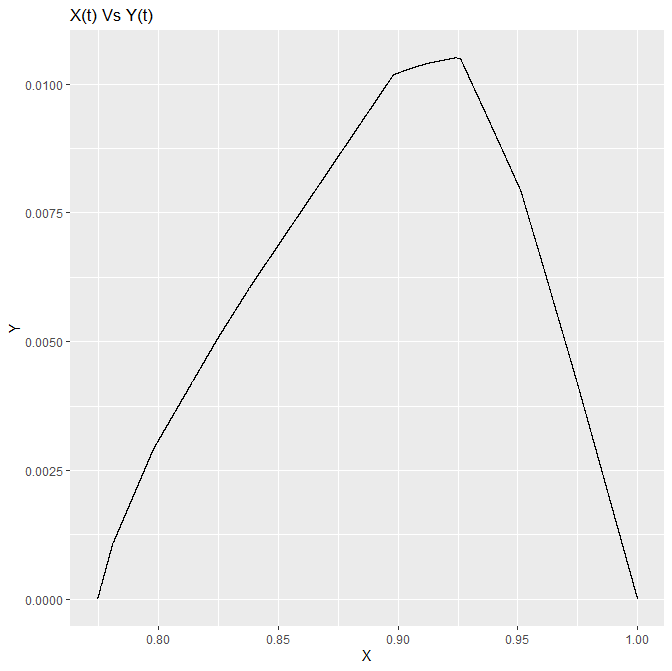}
		\end{subfigure}%
\begin{subfigure}{.6\textwidth}		
		\includegraphics[scale=0.25]{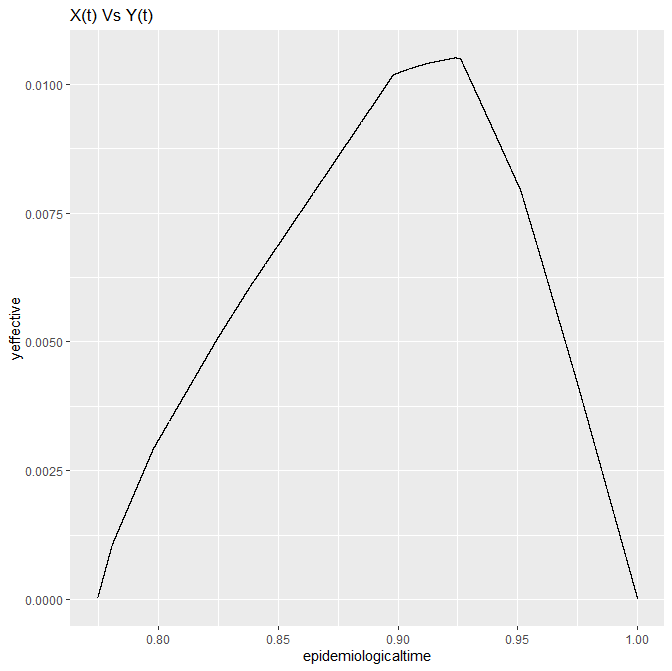}
				\end{subfigure}%
		\caption{Left panel: path of $(x(t), y(t))$ computed using Euler's scheme. Right panel: the same path computed using the integral equation. %Note, that to obtain the latter panel, only the integral equation involving $y(t)$ is used and there is no need to inverse the integral equation involving $x(t)$. 
		}
	\end{center}
\end{figure}

Figure 5 displays the trajectories of $y(t)$ as function of $x(t)$ corresponding to Figure 4. Thus, in epidemiological time, both methods are be quite accurate. However, in calendar time, when $\epsilon$ varies, the stability of the Euler scheme is questionable.   
\vspace{1em}

\paragraph{Deterministic SPSIR model.} Let us now consider the deterministic counterpart of the above analysis, based on the algorithms $iii)$ and $iv)$ above. Figure 6 (resp. Figure 7) are the analog of Figures 2 and 4 (resp. Figures 3 and 5).
\begin{figure}[H]
	\begin{center}
			\includegraphics[scale=0.3]{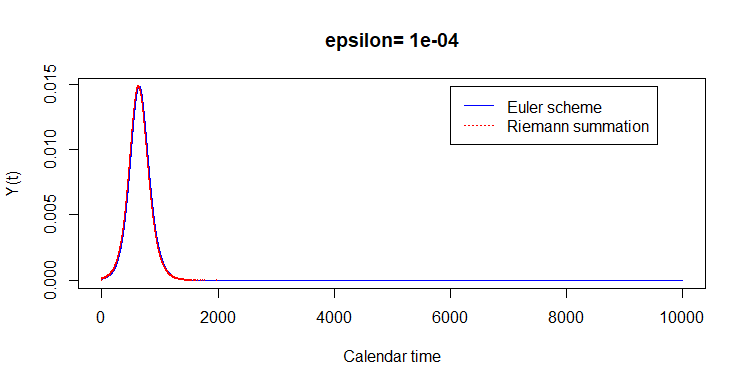}	
			\includegraphics[scale=0.3]{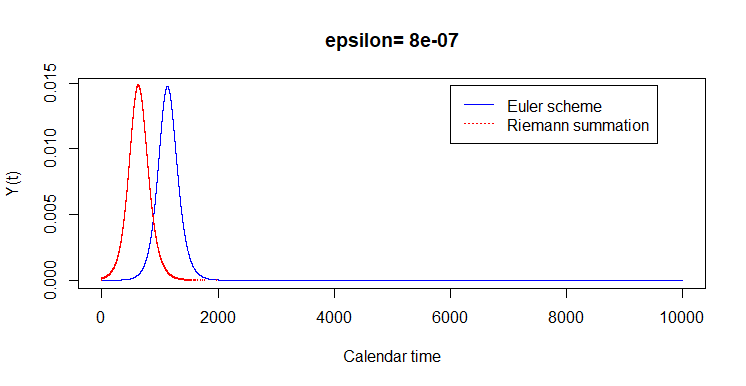}
		\caption{Upper panel: comparison of the path of $y(t)$ using the two approximation schemes, with $\epsilon=10^{-4}$. Lower panel: comparison of the path of $y(t)$ using the two approximation schemes, with $\epsilon=10^{-8}$. The other parameters have been set to their previous values. In both figures, the paths obtained with Riemann summation remains nearly unchanged, but those obtained by Euler's scheme differ significantly for the two values of $\epsilon$. %Note, that to obtain the latter panel, only the integral equation involving $y(t)$ is used and there is no need to inverse the integral equation involving $x(t)$. 
		}
	\end{center}
\end{figure}

\begin{figure}[H]
	\begin{center}
		\begin{subfigure}{.6\textwidth}
			\centering	
			\includegraphics[scale=0.25]{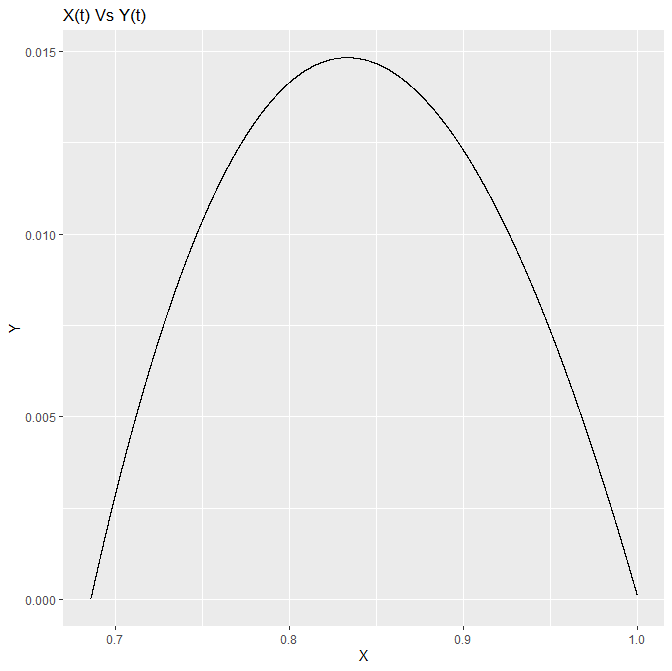}
		\end{subfigure}%
		\begin{subfigure}{.6\textwidth}		
			\includegraphics[scale=0.25]{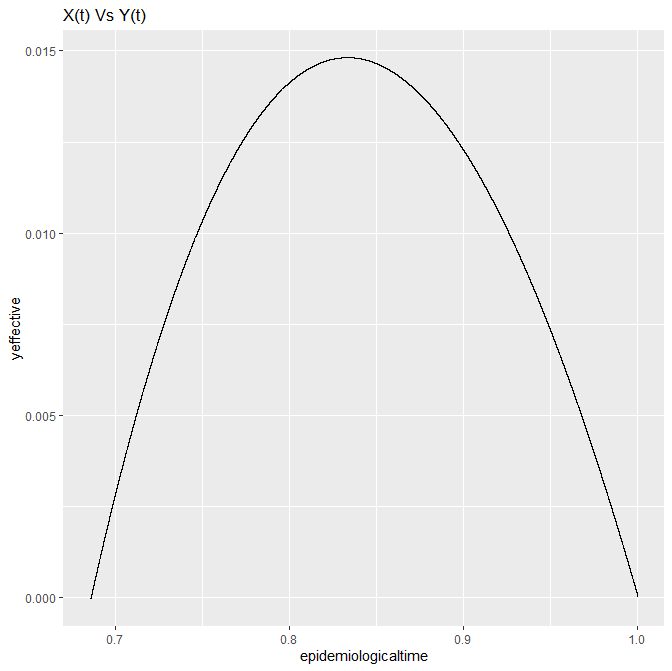}
		\end{subfigure}%
		\caption{Left panel: path of $(x(t), y(t))$ computed using Euler's scheme. Right panel: the same path computed using the integral equation. %Note, that to obtain the latter panel, only the integral equation involving $y(t)$ is used and there is no need to inverse the integral equation involving $x(t)$. 
		}
	\end{center}
\end{figure}
Again, despite providing unreliable paths for both $x(t)$ and $y(t)$, the Euler's scheme still provides reliable parametric equation $(x(t), y(t))$, that is the progression of the epidemic in epidemiological time (see Figure 7). 

\paragraph{Reproductive numbers. } We also compute the basic and effective reproductive numbers during the evolution of the epidemic by formulas \eqref{definition} and \eqref{effective}. We compute $x(t)$ and $y(t)$ using the algorithm $i)$, and plot both reproduction numbers in epidemiological time. They are displayed on Figure 8. 
\begin{figure}[H]
	\label{R0}
	\centering
	\includegraphics[scale=0.35]{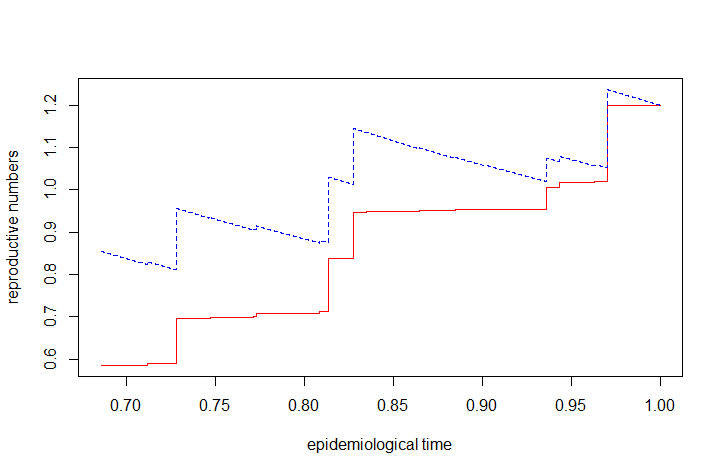}
	\caption{Basic (red full line) and effective (blue dashed line) reproductive numbers as a function of the epidemiological time $x(t)$, which decreases from $1-\epsilon$ for $t=0$ to the final size $x(\infty)$ at $t=\infty$.}
\end{figure}
This figure shows the variability in time of the $R_0$. In particular, the jumps of the reproduction numbers in the plots are due to the fact that the trajectory of $A(\cdot)$ can feature abrupt change (see Figure 2). Such non-smoothness can also be observed in Figure 5, which is based on the same stochastic transmission function $A(\cdot)$. This effect does not exist in Figure 7, which is based on a smooth, deterministic function $A(x)=\alpha x$.

\subsection{Monte-Carlo study}
The figures in Section 4.1 have been constructed from one trajectory of the gamma bridge transmission process. Let us now discuss the consequences of the uncertainty on the transmission function. By replicating the simulations, we can derive, by Monte-Carlo, the distribution of:
\begin{itemize}
	\item Reproductive numbers $R_0(t)$, and $R_e(t)$ at 20 days and 30 days;
	\item the time of the peak $\tau$, defined by the condition $A[X(\tau)]=c$, or equivalently $X(\tau)=A^{-1}(c)$. Using (2.4), it is given by:
	$$
	\tau=\displaystyle  \Int^{x(0)}_{A^{-1}(c)} \frac{dv}{A(v) [ v-1- c \Int^v_{x(0)} \Frac{du}{A(u)}]}.
	$$
	\item the size of the peak $Y(\tau)$, which can be computed through:
	$$
	Y(\tau)=1-A^{-1}(c)+ c \Int_{x(0)}^{A^{-1}(c)} \Frac{du}{A(u)}.
	$$
	\item the time at which the proportion of infected people is larger than 1$\%$ (called days of saturation henceforth), to see the compatibility with the number of available beds. 
\end{itemize}
We use the same parameter values as in Figure 2, and simulate 200 realizations of the epidemic process. The (marginal) distributions of the above epidemiological indicators are provided in Figure 9.

\begin{figure}[H]
	\centering
	\includegraphics[scale=0.6]{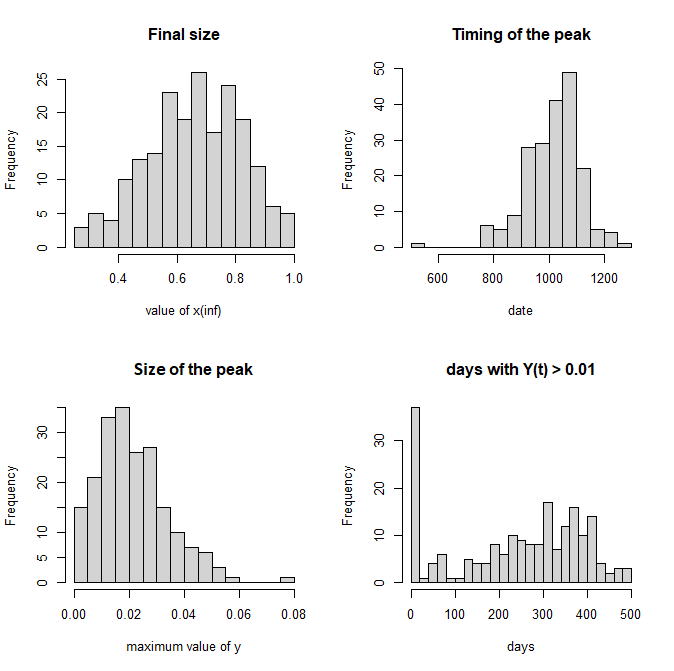}
	\caption{Histograms of the final size (upper left panel), the time of the peak (upper right panel), the size of the peak (lower left panel), as well as the days of saturation (lower right panel). }
\end{figure}
All the four epidemiological indicators vary greatly across the different scenarios. In particular, in the Southeast panel, for some scenarios, $y(t)$ never exceeds 1$\%$, which explains a large probability of observing value 0. To see how these indicators are related to each other, we also report their pairwise scatterplots: 
\begin{figure}[H]
	\centering
	\includegraphics[scale=0.6]{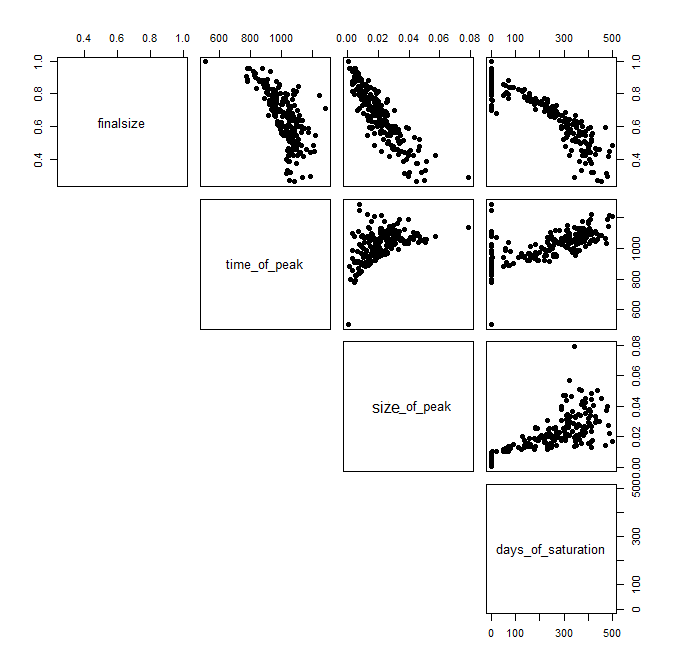}
	\caption{Pairwise scatter plots between the four epidemiological indicators considered in Figure 9.}
\end{figure}
For instance, in the upper right corner, we observe a negative dependence between the final size and number of days of saturation. We also observe the positive dependence between the time and size of the peak. Longer the explosion episode, higher the size of the peak.

%Let us now compare the discretization using the closed form solution, as well as the one involving crude Euler scheme. We focus on the computation of the herd immunity level, that is the limit $x(\infty)$ given by the two approaches. 

%We specify the stochastic function $A(\cdot)$ as $A(x)=xB(x)$, where $B(\cdot)$ is the trajectory of a gamma process, hence both functions $B$ and $A$ are increasing. Further, we assume the value of $B(1)$ known and non stochastic. As a consequence, for $B(x)/B(1)$ has beta distribution $Beta(x, 1-x)$. For a fixed $N$, 
%We first compute $x(\infty)$ using the closed form solution, that is, by solving the integral integration in $v$:$$v=1-c \int_{1-\epsilon}^v \frac{\mathrm{d}u}{A(u)}$$
%where $x(0)=1-\epsilon$ for a small positive constant $\epsilon$. As the integral in the above equation cannot be further simplified, we approximate it using Riemann sums. 

%Since $A$ is increasing, the integral $\int_{1-\epsilon}^v \frac{\mathrm{d}u}{A(u)}$ is under-estimated by its right sum. 
%\textbf{Remark 4~:} It would also be possible to assume a gamma process, or an ICIR process, for $1/A$ instead of $A$, taking into account the change of monotonicity. This may simplify computations or simulations due to the form of integral equation (2.2).

\section{Observational uncertainties}
\setcounter{equation}{0}\def\theequation{5.\arabic{equation}}

The stochastic SIR model introduced in Section 3 concerns theoretical cross-sectional probabilities corresponding to a virtual population of infinite size. However, in practice, the population size is generally large, but finite, and the observed cross-sectional frequencies differ from their theoretical counterparts. The aim of this section is to analyze in detail the distribution of observed counts when the underlying mechanistic model is a stochastic SIR. In this section, we first describe how the observations are related to the underlying SIR model with stochastic transition. This leads to a new type of nonlinear state space model. 

%Such a model is usually estimated by simulation based methods as method of simulated moments, or indirect inference. These approaches are based on artificial paths simulated from the SIR model, or more precisely from a time discretized version of this model with a small time unit. However in this nonlinear dynamic framework, these approximations are not always reliable in particular for the long run properties of the evolution of the disease. %This feature is discussed in detail in Appendix C on Euler discretization.

This nonlinear state space representation is appropriate for simulating artificial paths of observations from the SIR model (then also to apply simulation based estimation methods). We will provide such simulations of cross-sectional frequencies that can be compared with the simulations of their theoretical counterparts corresponding to a population of infinite size. This allows for comparing the effects of the two types of uncertainties: the intrinsic uncertainty due to the stochastic transmission function and the observational, sampling uncertainty due to the finiteness of the population size. We will study how their relative magnitude vary during the evolution of the epidemic.

%To circumvent this potential difficulty, we propose to estimate the main parameters by marginal likelihood maximization, using the simple dynamics of the proportion of infected individuals written in the transformed (inverse) time scale defined by the proportion of susceptible individuals.

\subsection{The model}

The observations are generally available in discrete time, i.e. daily, and aggregated over individuals to get time series of counts. Let us denote $[I_{i,t}, t \in (0, \infty)]$ the continuous time series of state indicators for individual $i, i=1,...,N$. Thus $I_{i,t}=j$, $j=1,2,3$, if individual $i$ is in state $j$ at date $t$, where $j=1=S, j=2=I, j=3=R$.

%. The model is defined as follows~:\vspace{1em}

\begin{Ass}
	\label{exchangeable}
$i)$ The population size $N$ is time independent.

$ii)$ The continuous time individual histories $(I_{i,t}, t \text{ varying })$, $i=1,...,N$ are independent, identically distributed. 

$iii)$ The dynamics of an individual history is driven by a stochastic intensity deduced from system (3.3). 
\end{Ass}

Assumption \ref{exchangeable} $ii)$, $iii)$ are standard [see e.g. Gourieroux, Jasiak (2020a, b)]. They mean that the population of individuals is "homogeneous", that is, these individuals are exchangeable and their histories are independent. %Then the individual observations can be aggregated\footnote{Note that these cross-sectional counts are not sufficient summaries of the individual histories. In such a model with stochastic transmission, aggregate flow counts are also not sufficient [see Breto et al. (2009) for the use of flow counts].}, and the multinomial distribution (conditional on $X,Y,Z)$ is a direct consequence of such i.i.d. assumptions of individual histories (given $X,Y,Z$).\vspace{1em}

Under Assumptions \ref{stochastictransmission}, \ref{exchangeable}, we can aggregate the individual histories to get daily counts. We denote by $N_{jk}(t)$ the number of individuals transiting from state $j$ at date $t-1$ to state $k$ at date $t$. Note that by definition of the states we have: $N_{jk}(t)=0, \forall k<j$. We also denote the marginal counts by $N_k(t)$, that are the number of individuals in state $k$ at date $t$. By definition, we have:
\begin{align}
	N_k(t)&= \sum_{j=1}^3 N_{jk}(t), \qquad \forall k \in \{1,2,3\}, \\
	N_j(t-1)&= \sum_{k=1}^3 N_{jk}(t), \qquad \forall j \in \{1,2,3\}.
\end{align}

%\begin{Ass} 
%	\label{poissonapproximation}
%	Parameters $\alpha, c$ are small. More precisely, $\alpha, c$ tend to zero when $N$ tends to infinity, such that $\lim_{N \to \infty} \alpha N= \alpha^*>0$, $\lim_{N \to \infty} c N= \gamma>0$. 
%\end{Ass}

Let us now introduce sufficient conditions on the fundamental $c, N, A(\cdot)$ of the SIR model to allow for Poisson approximation. 
\begin{Ass} 
	\label{poissonapproximation}
	Parameters $c$ and the process $A(x)/x$, indexed by $x$, are small. More precisely, $c$ and $A(x)/x$ tend to zero, when the population size $N$ tends to infinity, such that $\lim_{N \to \infty} c N= \alpha^*>0$, and process $N A(x)/x$ converges weakly to $\delta(x)>0$, which is a stochastic process indexed by $x$, whose distribution is independent of $N$.  
\end{Ass}
%Note that since $A(1)$ is stochastic, the limit $\alpha^*$ can also be stochastic (but will become ``recoverable" at the outset of the epidemic). 
%Let us discuss specifications of the stochastic transmision function that satisfiy Assumption 4.  If $A(\cdot)$ is a gamma process with degree of freedom parameter $\gamma$ and scale parameter $\theta$, then the condition $\lim_{N \to \infty} c A(x)/x= \delta(x)>0$ is satisfied, if $\lim_{N \to \infty} c \theta$ is finite and $\gamma$ is independent of $N$. 

Let us discuss some sufficient conditions for Assumption 4 to hold. We first write any stochastic transmission process $A(\cdot)$ as $A(x)= A(1) H(x)$, with process $H(\cdot)$ increasing and satisfies $H(1)=1$. Then we assume that $A(1)$ and $H(\cdot)$ are independent. Moreover, we assume that only the distribution of $A(1)$ depends on $N$, but not that of $H(\cdot)$.  Then a sufficient condition for $\lim_{N \to \infty} N A(x)/x$ to be positive and finite is that $\lim_{N \to \infty} N A(1)$ is positive and finite. Then process $H(\cdot)$ can be any increasing process that maps $[0, 1]$ to itself. For instance, the following two specifications are compatible with the Examples 4 and 6 given in Section 3.3. 
\begin{itemize}
	\item We can choose $H(\cdot)$ to be a L\'evy bridge, that is a L\'evy process on $[0, 1]$ given that its terminal value is equal to 1. Note that process $A(\cdot)$ typically is not a L\'evy process. Indeed, it is L\'evy if and only if $A(1)$ is gamma distributed, that is, if and only if $A(\cdot)$ is a L\'evy process, which is the model considered in Example 4 and used in the simulation Section 4. In this case, for any given $x \in [0, 1]$, the distribution of $H(x)$ is Beta($\theta x, \theta (1-x))$, where $\theta$ is the degree of freedom parameter of the gamma process.  
	\item  alternatively, we can specify $H(\cdot)$ as the time changed L\'evy process defined in Example 6, that is $H(x)=\exp(vB_{K_v(x)})$, see equation (3.6). 
\end{itemize}

\begin{Ass}
	\label{markov}
	Process $A(\cdot)$ is Markov in reverse time.
\end{Ass}
For instance, this Markov assumption is satisfied in Examples 4 and 6. 

\begin{prop}
	Under Assumptions \ref{stochastictransmission}, \ref{exchangeable}, \ref{poissonapproximation} and \ref{markov}, 
	\begin{enumerate}[$i)$]
		\item Process $N(t):=[N_1(t), N_2(t), N_3(t)=N-N_1(t)-N_2(t)]'$ is a Markov process, conditional on stochastic transmission process $(A(x))_{x \in (0, 1)}$.
		\item This process is also Markov conditional on the information set $[A(x), x>X(t)]$
		\item The transition matrix between $t-1$ and $t$ conditional on stochastic transmission process $(A(x))_{x \in (0, 1)}$ can be written as:
		$$
		P_t= \begin{bmatrix}
			1-p_{12}(t) &  p_{12}(t) & o(\frac{1}{N}) \\
			0 & 1-p_{23}(t) & p_{23}(t) \\
			0& 0& 1 
		\end{bmatrix},
		$$
	%	where $p_{12}(t)$, $p_{23}(t)$ are of order $\frac{1}{N}$, and $p_{13}(t)$ is negligible with respect to $\frac{1}{N}$. 
	where \begin{align*}
p_{12}(t)&= 1-\exp\Big(-\int_{t-1}^t \frac{A[X(s)]}{X(s)} Y(s) \mathrm{d}s \Big),\\
\text{and  }p_{23}(t)&=1-\exp(-c),
\end{align*}
are of order $\frac{1}{N}$, and $p_{13}(t)$ is negligible with respect to $\frac{1}{N}$.
		\item The count $N_{13}(t)$ is negligible with respect to counts $N_{12}(t), N_{23}(t)$.
		\item Conditional on the information available at date $t-1$, and for large $N$, $N_{12}(t)$ and $N_{23}(t)$ are independent, with approximate Poisson distributions:
		\begin{equation}
			\label{conditionalpoissondistribution}
		N_{12}(t) \sim \mathcal{P}(N_{1}(t-1)p_{12}(t)), \qquad 	N_{23}(t) \sim \mathcal{P}(N_{2}(t-1)p_{23}(t)).
	\end{equation}
	\end{enumerate}
\end{prop}

\begin{proof}
	$i)$ is a direct consequence of the stochastic differential system, that becomes deterministic given $A(\cdot)$. 
	
$ii)$ is also a consequence of this deterministic system conditional on $A(\cdot)$. This implies that $N(t)$ depends on the information available at $t-1$ by means of just $N_{1}(t-1)$ and $N_{2}(t-1)$ that are complicated summary statistics of the transmission process up to stochastic, epidemiological time $X(t)$.

$iii)$ The form of the transition matrix is a consequence of Assumption A.3. It is obtained by integrating the effect of the stochastic intensity between $t-1$ and $t$, that is of integrating $A(\cdot)$ between $x(1-\epsilon)$ and $x(t-\cdot)$ conditional on the path of $A(\cdot)$ from $1$ to $X(t-\cdot)$. The fact that $p_{13}(t)$ becomes negligible simply means that a susceptible individual cannot be infected, then recover in just one day. 

$iv)$ In particular, we have $N_{13}(t)= 0, \forall t$. Then we have:
\begin{align}
	N_{12}(t)&=N_1(t-1)-N_1(t), \label{increment1}\\
	N_{23}(t)&=N_3(t)-N_3(t-1). \label{increment2}
\end{align}
Therefore we have the same information in the sequence of counts $[N_{jk}(t), t \text{ varying}]'$, and in the sequence of cross-sectional counts $[N_k(t), t \text{ varying}]'$

$v)$ This is a consequence of aggregation of qualitative histories: The counts corresponding to two different lagged status are independent. Therefore $[N_{11}(t), N_{12}(t), N_{13}(t)]$ are independent of the $[N_{21}(t), N_{22}(t), N_{23}(t)]$ given the past. Moreover, they are such that $N_{12}(t) \sim \mathcal{B}(N_{1}(t-1)p_{12}(t))$, $N_{23}(t) \sim \mathcal{B}(N_{2}(t-1)p_{23}(t))$ given the past. Then the result follows from Assumption A.4 that allows for the Poisson approximation of the binomial distribution. 

%\textbf{Comment by Yang: Alternative proof without relying on Assumption 4: assuming that $p_{12}(t)$ is small, then the variance of the conditional binomial distribution is $N_1(t-1) p_{12}(t) (1- p_{12}(t)) \approx N_1(t-1) p_{12}(t) $. Hence we have approximately a Poisson distribution, whose parameter is (roughly) proportional to the population size $N$.}

%\textbf{La raison pour laquelle je pense qu'il faut éviter des hypothèses du tyme $NA(x)$ et $cN$ constant est qu'elles impliquent que $c$ et $A(x)$ tendent vers zéro. Du coup $y(t)$ tend vers zéro aussi pour $t$ fixé. Si l'on regarde l'expression de $p_{1,2}(t)$, on aperçoit que ceci tend vers zéro plus vide que $1/N$, rendant l'approximation Poissonienne non standard. Le problème vient du fait que en tendant $c$ et $A(x)$ vers zéro, on a changé aussi les proportions $x(t), y(t)$, c'est à dire la dynamique dans le temps épidemiologique. Or le SIR de base est utilisé à la fois pour les comptage et pour les proportions et derrière il y a cette idée d'invariance par échelle.

%Il faut sans doute aussi parler de la valeur initiale: dans Sections 2, 3, 4, on travaillait sur les proportions et on suppose $y(0)=\epsilon$. Maintenant, comme on passe aux comptage, il faut une condition initiale pour $N_2(0)$. Cela peut être $N_2(0)=1$, ou $N_2(0) \approx N\epsilon$, avec $\epsilon$ fixé ne dépendant pas de $N$. Ces deux hypothèses ne sont pas la même}
\end{proof}

By using (5.2), we deduce the following corollary:
\begin{Corollary}
	The conditional distribution of $\hat{x}(t)= N_1(t)/N$, $\hat{y}(t)=N_2(t)/N$ given is deduced by noting that:
	\begin{align*}
		 \hat{x}(t)&=\hat{x}(t-1)-\frac{N_{12}(t)}{N}\\
		 \hat{y}(t)&=\hat{y}(t-1)+ \frac{N_{12}(t)}{N}-\frac{N_{23}(t)}{N}
	\end{align*}
where $[N_{12}(t), N_{23}(t)]'$ follow the distribution in Proposition 6 $v)$. 
\end{Corollary}
%\textbf{Comment by Yang: Corollary 4 sounds like a direct consequence of equations \eqref{increment1} and \eqref{increment2}. Is it really necessary?}
\vspace{1em}
%observations are the frequencies $\hat{x} (t), \hat{y}(t), \hat{z} (t), t=1,\ldots, T.$

%iii) Conditional on the underlying processes : $X(t), Y(t), Z(t), t \in (0,\infty)$, the observations $[\hat{x} (t), \hat{y}(t), \hat{z}(t)], t=1,\ldots,T.$ are such that: $N[\hat{x} (t), \hat{y}(t),\hat{z}(t)]$ follows a multinomial distribution $\mathcal{M}[N;X(t),Y(t),Z(t)]$, for any $t$.

%iv) The joint process $(X,Y,Z)$ satisfies the SIR model with stochastic transmission $A$, defined in (3.3), with a parametric distribution for process $A$.\vspace{1em}

\textbf{Remark 8: }Under Assumption 4 that allows for the Poisson approximation of the binomial distribution, the observed frequencies are not necessary good approximations of their theoretical counterparts. Loosely speaking, we have $\hat{x}(t)-x(t)= o (\hat{x}(t))$, if $\gamma(x(t))$ is very large for $N$ tending to infinity. Otherwise, the approximation is not necessarily very accurate. 

\vspace{1em}

Under Assumptions 3, 4 and 5, the model has a non-linear state-space representation with four types of variables, from the highest to the lowest layer~:

$i) A(x), x \in [0, 1]$ is the continuous state latent transmission function;

$ii) X(t), Y(t), Z(t), t>0$ are the continuous time latent probabilities;

$iii) N_{12}(t), N_{23}(t), t=1,...,N$ the observed cross-counts;

$iv) \hat{x} (t), \hat{y}(t), \hat{z}(t), t=1,...,N$ are the discrete time, observed frequencies;
\vspace{1em}

Layers $i)$ and $ii)$ contain state variables: those in $i)$ are stochastic, whereas those in layer $ii)$ are linked deterministically to layer $i)$; Layers $iii)$ and $iv)$ contain measurement variables: those in $iii)$ are linked stochastically to those in layer $ii)$, and those in layer $iv)$ are linked to $iii)$ in a deterministic way. Thus two types of uncertainties are considered, that are the cross-sectional uncertainty (or demographic uncertainty) in the measurement equation $iii)$, and the transmission uncertainty (or environmental uncertainty) in the first layer of nonlinear state equations.

\subsection{Intrinsic versus observational uncertainty}
 
 Let us now repeat the simulation exercise of Section 4, but now for the count variables. Corollary 4 shows that there is the same information in the cross counts and in the marginal counts (a property specific to the SIR model). Therefore we will simulate the marginal frequencies directly.
 
% \textbf{comment by Yang: I would say it is easier to simulate the cross counts using equation \eqref{conditionalpoissondistribution}, then the marginal counts (and hence the marginal frequencies) are deduced by equations \eqref{increment1} and \eqref{increment2}. }
 
 %\textbf{Aussi, je note que l'hypothèse est légèrement différent de ton papier avec Nour, où les intensités sont endogènes et dépendent des fréquences empiriques (à la GARCH). Ici les intensités de transition sont exogènes. Je n'ai pas de problème avec cette hypothèse d'exogéneité, car ça permet de mieux mesurer les deux types de stochasticité. Mais il faut sans doute dire quelque part qu'une autre approche est aussi possible, car l'autre approche est plus classique dans la littérature de probabilité théorique (il semble que ce papier de Daniels est une référence: "THE DISTRIBUTION OF THE TOTAL SIZE OF AN EPIDEMIC")}
 
 We focus on the first $T=400$ of the epidemic, for a population of size $N=50000$. We first simulate one trajectory of $(X(t), Y(t))$, and use equation \eqref{conditionalpoissondistribution} to simulate the trajectory of $N_2(t)$. The following figure plots the time series of counts of infectious individuals. The parameters of the SIR model are the same as that used in Figure 2, except that $R_0={\alpha}/{c}=2$.
  
 \begin{figure}[H]
 	\centering
 	\includegraphics[scale=0.5]{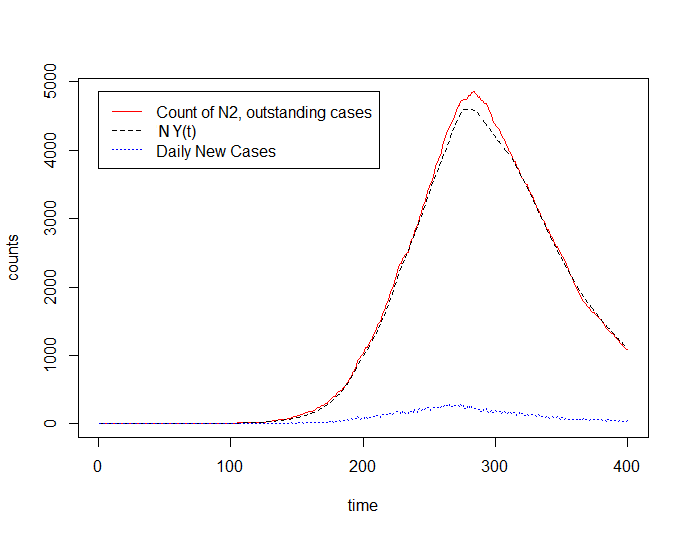}
 	\caption{Full red line: Time series of $N_2(t)$, that is the outstanding infectious cases. Black dashed line: expected number of outstanding cases given the trajectory of $Y$, that is the product between the population size $N$ and the value of $Y(t)$; blue dotted line: Time series of new infection cases. }
 \end{figure}
The observational uncertainty is quite small, even for this population of rather limited size.\footnote{If we increase the population size $N$, the discrepancy between the two curves becomes smaller and smaller due to the law of large numbers applied in the cross-section. Hence these plots are omitted. }On the other hand, the peak of the epidemic is attained at around $t=300$, thus at $T=400$, $X(400)$ is significantly different from $X(0)$ (in the above simulation we have $X(400) \approx 0.75$) and the effect of stochasticity of $A(\cdot)$ is expected to be quite important. To compare these two types of uncertainties, we replicate the same simulation experiments 50 times and compute the histogram of the 
sampling distribution of $N_2(T)=N_2(400)$ and $N Y(400)$. The former accounts for both types of uncertainty, whereas the latter only accounts for the intrinsic uncertainty. 
\begin{figure}[H]
	\centering
	\includegraphics[scale=0.5]{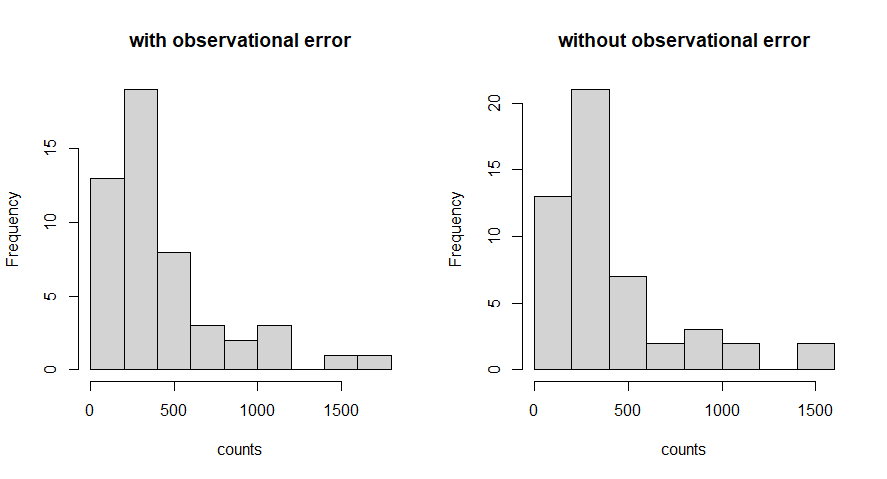}
	\caption{Left panel: histogram of the final count $N_2(400)$. Right panel: histogram of $N Y(400)$. }
\end{figure}
Both histograms are quite similar, hence the intrinsic uncertainty is much more important than the sampling uncertainty. 

Let us now repeat the same simulation, but by diminishing significantly $R_0$ from 2 to 1.2. This means that the epidemic evolves much more slowly, and hence to get comparable counts, we multiply the size of the population used in the previous two figures by 100 to arrive at $N=5,000,000$,  (which is roughly the size of large Metropolitan regions such as Boston, Toronto, etc). The following figure is the analog of Figures 11 and 12.

\begin{figure}[H]
	\centering
	\includegraphics[scale=0.4]{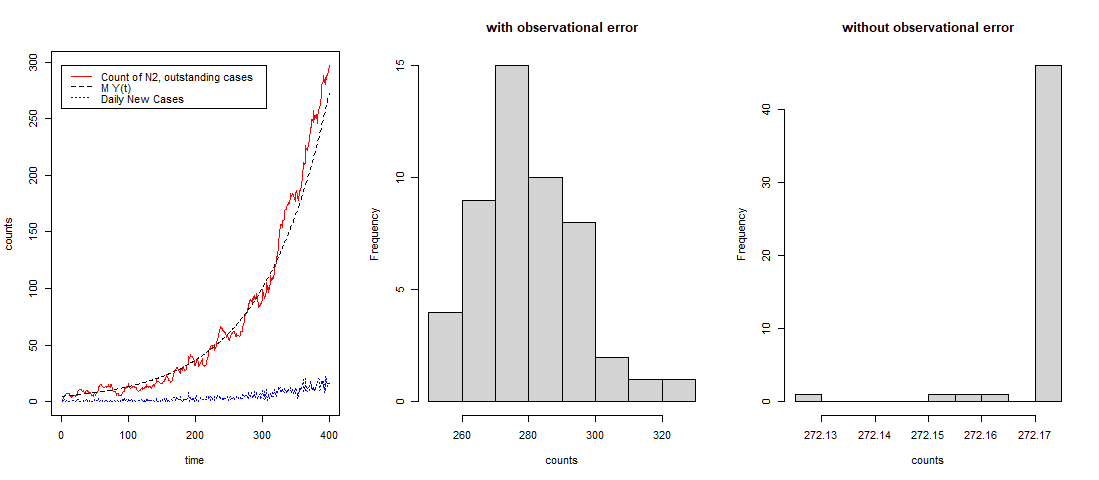}
	\caption{Left panel: Simulated trajectories of $N_2(t), NY(t)$ for $t=1,...,T$. Central panel: Histogram of the final count $N_2(T)$ for 50 realizations of $A(\cdot)$. Right Panel: Histogram of $N Y(T)$. }
\end{figure}
In this example, at $T=400$, the epidemic is still at an early stage, with $X(400) \approx 0.03 \%$ across various scenarios. The right panel shows that the variation of $Y(400)$ across different epidemiological scenarios is nearly negligible, which is expected, since process $A(\cdot)$ is continuous at 1, we have $A(t) \approx \alpha$ for any $t=1,...,T$, which does not vary across different scenarios. Then the observational uncertainty, characterized by the histogram in the central panel, becomes dominant. 

Therefore, the relative importance of the sampling and intrinsic uncertainties depends on the stage of the epidemic. At the early stage, the former is more important; as we advance to later stages, the intrinsic becomes dominant. Moreover, the larger the population, the smaller the relative importance of the sampling uncertainty. %This latter, will also depend on the variation of the $A(\cdot)$, which in our gamma bridge model is controlled by parameter $\theta$. 
\section{Concluding remarks}

The aim of this paper was to increase the flexibility of the classic SIR model by introducing a stochastic nonlinear transmission rate. This approach is a structural way of introducing intrinsic uncertainty in the SIR model. It is appropriate to capture the mover-stayer phenomenon, to  account for the observed variability of the estimated reproductive numbers, to perform risk analysis on the required number of intensive care unit beds, or to disentangle observational and intrinsic (i.e. systemic) risks.

The modelling approach has used exact solutions existing in the classic SIR model. It can be immediately extended to other epidemiological models with exact solutions, as the  Susceptible-Infected-Recovered-Deceased (SIRD) model, the SIR models with random time of recovering [Guan et al. (2019)], to SIR type model in which the transmission rate is adjusted for the proportion of immunized people [the first equation in (2.1) is replaced by $dx(t)/dt = - \alpha x(t) y(t)/[x(t) + y(t)]$, and the second equation modified accordingly] [see Bailey (1980), Bohner et al. (2019)]. 

%The perfect knowledge of the properties of such extended SIR models and of their discrete time analogues are required before using these epidemiological models to evaluate the economic impacts of an epidemic [see e.g. Niepelt, Gonzales-Eiras (2020) for discussion].

The existence of a nonlinear state representation will facilitate the estimation of a stochastic SIR model. The description of such estimation method, the appropriate filtering algorithms, and their properties are left for future research [see however the discussion of the use of Kalman filter in models without intrinsic uncertainty in Breto et al. (2009) and Gourieroux, Jasiak (2020a,b)].

%\newpage
\appendices
 \appendix
\section{Solution of the SPSIR model}

\setcounter{equation}{0}\def\theequation{a.\arabic{equation}}
To solve the SPSIR model, we focus on the two first differential equations, i.e. the master subsystem~:

\begin{eqnarray}
       \Frac{dx(t)}{dt} &=& -a[x(t)] y(t), \\
       \Frac{dy(t)}{dt} &=& a [x(t)] y(t) - c y(t) = \{ a[x(t)] - c\} y (t),
     \end{eqnarray}

\noindent where function $a$ is strictly positive on $]0,1]$, with $a(0)=0.$\vspace{1em}

\subsection{Equation for $x(t)$} 

From the first equation (a.1), we get~:

$$
y(t) = - \Frac{dx (t)}{dt} \Frac{1}{a[x(t)]},
$$

\noindent and by differencing we get~:

$$
\Frac{dy (t)}{dt} = - \Frac{d^2 x(t)}{dt^2} \Frac{1}{a[x(t)]} + \Frac{\Frac{da}{dx} [x(t)]}{a[x (t)]^2} \left(\Frac{dx(t)}{dt}\right)^2.
$$

Then, by plugging in these expressions in eq. (a.2), we get~:

$$
-\Frac{d^2 x(t)}{dt^2} \Frac{1}{a[x(t)]} + \Frac{\Frac{da}{dx} [x(t)]}{a[x(t)]^2} \left( \Frac{dx (t)}{dt}\right)^2 = - \{ a[x(t)]-c\} \Frac{dx(t)}{dt} \Frac{1}{a[x(t)]},
$$

\noindent or equivalently,

\begin{equation}
  -\Frac{d^2 x(t)}{dt^2} + \Frac{d\log a[x(t)]}{dx} \left[ \Frac{dx(t)}{dt}\right]^2 + \{a [x(t)]-c\} \Frac{dx(t)}{dt} = 0.
\end{equation}

\subsection{Time deformation} 
By the continuity and strict monotonicity of function $x(t)$, we deduce that this function is invertible. Its inverse $u(x)$ can be seen as a time deformation and satisfies the equality~:

\begin{equation}
  u[x(t)] = t, \forall t.
\end{equation}

A first differencing of this equality leads to~:

$$
  \Frac{du [x(t)]}{dx} \Frac{dx(t)}{dt} = 1,
$$

\noindent or, equivalently,

\begin{equation}
  \Frac{dx(t)}{dt} = 1/ \Frac{du [x(t)]}{dx} .
\end{equation}

A second differentiation of (a.4) leads to~:

$$
\Frac{d^2 u [x(t)]}{dx^2} \left( \Frac{dx (t)}{dt}\right)^2 + \Frac{du[x(t)]}{dx}  \Frac{d^2 x (t)}{dt^2} = 0,
$$

\noindent or, by (a.5)~:

\begin{equation}
  \Frac{d^2 x(t)}{dt^2} = - \Frac{d^2 u [x(t)]}{dx^2} / \left[ \Frac{du [x(t)]}{dx}\right]^3.
\end{equation}

\subsection{Equation for the time deformation}
By using (a.5)-(a.6), we deduce the differential equation satisfied by time deformation $u(x).$ We get~:

$$\Frac{d^2 u(x)}{dx^2} / \left[ \Frac{du(x)}{dx}\right]^3 + \Frac{d \log a (x)}{dx} / \left[ \Frac{du(x)}{dx} \right]^2 + [a (x) - c] / \Frac{du (x)}{dx} = 0,$$

\noindent or

\begin{eqnarray}
\Frac{d^2 u(x)}{dx^2} / \left[ \Frac{du(x)}{dx}\right]^2 &+& \Frac{d \log a (x)}{dx} / \Frac{du(x)}{dx} + a (x) - c  = 0.
\end{eqnarray}

\subsection{Solution of equation (a.7)}

Let us denote~:

\begin{equation}
  v(x) = \Frac{du}{dx} (x).
\end{equation}

Equation (a.7) becomes~:

\begin{equation}
  \Frac{1}{v^2(x)} \Frac{dv (x)}{dx} + \Frac{d\log a(x)}{dx} \Frac{1}{v(x)} + a(x) - c = 0.
\end{equation}

This is a Riccati equation with functional coefficients [Harko et al. (2014)].

Let us denote~:

\begin{equation}
  w(x) = 1/v(x).
\end{equation}

We get~:

\begin{equation}
  -\Frac{dw(x)}{dx} + \Frac{d\log a(x)}{dx} w(x) + a(x) - c = 0.
\end{equation}

This differential equation is linear of order 1. As usual it is solved in two steps.\vspace{1em}

i) First we consider the associated homogeneous equation~:

$$
\begin{array}{l}
- \Frac{dw(x)}{dx} + \Frac{d\log a(x)}{dx} w (x) = 0, \\ \\
\Longleftrightarrow \Frac{d\log w(x)}{dx} = \Frac{d\log a(x)}{dx},
\end{array}
$$

whose solution is~:

$$
w(x) = K a(x),
$$

\noindent where $K$ is a constant.\vspace{1em}

ii) Then we apply the varying constant technique and look for a solution of (a.11) of the form~:

$$
w (x) = K(x) a(x).
$$

We get~:

\begin{eqnarray}
&-&\Frac{dK(x)}{dx} a(x) - K(x) \Frac{d a(x)}{dx} + \Frac{d \log a(x)}{dx} K(x) a(x) + a(x) - c = 0 \nonumber \\
&\Longleftrightarrow & - \Frac{dK(x)}{dx} a(x) + a(x) - c = 0 \nonumber \\
&\Longleftrightarrow &  \Frac{dK(x)}{dx} = 1- \Frac{c}{a(x)} \nonumber \\
&\Longleftrightarrow & K(x) = \Int^x_1 \left[1-\Frac{c}{a(u)} \right] du + K,
\end{eqnarray}

\noindent where $K$ is an integrating constant.\vspace{1em}

Therefore we get~:
\begin{eqnarray}
  w(x) = 1/v(x) & = & 1/\Frac{du}{dx} (x) = a(x) \{ \Int^x_1 \left[1-\Frac{c}{a(u)}\right] du + K\} \nonumber \\
  &=& a(x) \{ x-1 - c \Int^x_1 \Frac{du}{a(u)} + K\}.
\end{eqnarray}

The constant $K$ can be derived from the initial conditions for $x$ and $y$, $x(0), y(0),$ say. From eq. (a.1) written for $t=0$, we get~:

$$
\Frac{dx(0)}{dt} = 1/ \Frac{du}{dx} [x(0)] = w [x(0)] = -a [x(0)] y(0).
$$

Therefore by (a.13), we get~:

$$
w [x(0)] = -a[x(0)] y(0) = a[x(0)] \{ x (0) - 1-c \Int^{x (0)}_1 \Frac{du}{a(u)} + K\}
$$

\begin{equation}
  \Longleftrightarrow - y(0) = x(0) -1-c \Int^{x(0)}_1 \Frac{du}{a(u)} + K.
\end{equation}

This provides the constant $K$ and eq. (a.12) becomes~:

\begin{equation}
  w(x) = a(x) \{ x-x(0)-y(0) - c \Int^x_{x(0)} \Frac{du}{a(u)}\}.
\end{equation}

Since $\Frac{du}{dx} = \Frac{1}{w(x)}$, we deduce the expression of $u(x)$:

\begin{equation}
  u(x) - u [x(0)] = \Int^x_{x(0)} \Frac{dv}{a(v) \{ v-x(0) - y(0)-c \Int^v_{x(0)} \Frac{du}{a(u)}\}}.
\end{equation}

The result in Proposition 3 is obtained by replacing $x$ by $x(t)$ and noting that $u [x(t)] = t$.\vspace{1em}

\section{Condition for herd immunity}
 
\subsection{Proofs of Proposition 4 and Corollary 2}
\begin{proof}[Proof of $4 i)$]
Let us define the function:
\begin{equation}
\label{functiong}
%h_\epsilon(v):=
h(v):=v-x(0)-y(0)-c \int_{x(0)}^v \frac{\mathrm{d}u}{a(u)}=0,
\end{equation}
where argument $v$ varies between 0 and 1. %and we have introduced the notation $\epsilon=1-x(0)$.
We have: $h(x(0))=-y(0)<0$, and $h(0)=-x(0)-y(0)+c \int_0^{x(0)} \frac{\mathrm{d}u}{a(u)}$, with the convention that $h(0)=\infty$ if $\frac{1}{a(u)}$ is non integrable at $0$. Since $h(\cdot)$ is continuous for a given $x(0)$, there exists at least one solution $v$ between $0$ and $x(0)$, if and only if $-x(0)-y(0)+c \int_0^{x(0)} \frac{\mathrm{d}u}{a(u)}>0$. 

When at least one solution exists, the solution is unique over $]0,x(0)[$ since $\frac{\mathrm{d}h}{\mathrm{d}v}(v)=1-\frac{c}{a(v)}$ is increasing over $[0,1]$. %Hence we have proven assertion $4. ii)$
Hence we have proven assertion $4. i)$.
\end{proof}

\begin{proof}[Proof of Corollary 2]
By tending $x(0)$ towards $1$, we get equivalently: $$
c \int_0^1 \frac{\mathrm{d}u}{a(u)}>1,
$$
Hence we have proven Corollary 2.
\end{proof}

\begin{proof}[Proof of $4 ii)$]
Let us now show that, if the herd immunity is reached, it is necessarily reached at infinity. To this end, we proceed by contradiction. Assume that $x(T)=v^*$ for a finite $T$, then by replacing $t$ by $T$ and $x(t)$ by $v^*$ in equation \eqref{integralequation}, we deduce that the integral on the right hand side is necessarily finite. But by a first-order Taylor's expansion and assuming $a(\cdot)$ continuous\footnote{Here the continuity assumption is introduced merely for expository purpose. Indeed, since $a(\cdot)$ is assumed increasing, even if it is not continuous, we can always bound $1/a(v)$ from below by, $1/a(v^*)$ and the arguments below remain valid.}, we deduce that for $v$ close to $v^*$, the integrand is equivalent to:
$$
 \frac{1}{a(v^*)\Big[
 	1-\frac{c}{a(v^*)}\Big](v-v^*)},  
$$
which is non integrable at $v=v^*$. Thus we have a contradiction. %Thus we have $T=\infty$. Hence we have proven assertion $4. ii)$.
\end{proof}
\begin{proof}[Proof of $4 iii)$]
If inequality \eqref{herdimmunitycondition} is not satisfied, then \textit{a fortiori} $1/a(x)$ is integrable at zero since $\int_0^1 \frac{\mathrm{d}u}{a(u)}$ is finite. Hence $0$ cannot be a pole, either. In other words, the integrand at the right hand side of \eqref{integralequation} is finite for any value of $x(t)$ between $0$ and $x(0)$. Then the right hand side is upper bounded. Hence in equation \eqref{integralequation}, the admissible $t$ is upper bounded. Let us denote by $T$ its supremum. Then $x(T)$ is necessarily 0, since, otherwise, equation \eqref{integralequation} is still well defined beyond $T$, which contradicts the definition of supremum. %Hence we have proven assertion $4. iii)$.
\end{proof}
\subsection{Proof of Proposition 5}
Let us first define, for any $\epsilon \geq0$, the function: 
\begin{equation}
\label{functiong0}
h_\epsilon(v):=v-1-c \int_{1-\epsilon}^v \frac{\mathrm{d}u}{a(u)} 
\end{equation}
and consider the equation: $h_0(v)=0$. 
Clearly, $v=1$ is always a solution to this equation. By a similar proof as that of Proposition 3, we deduce that on $]0, 1[$, there is one and only one other solution $v_0$, if and only if $a(1)>c$.

$i)$ Let us first consider the case $a(1)>c$. In this case, the solution $(x_\epsilon(\infty))$ satisfies the inequality:
\begin{equation}
\label{boundedaway}
v_\epsilon < a^{-1}(c),
\end{equation}
where $a^{-1}$ is the inverse of function $a$. Indeed $a^{-1}(c)$ is the point where function $h_\epsilon(\cdot)$ attains its minimum over $0$ and $1$, for any value of $\epsilon$. In other words, this solution is uniformly away from 1 when $\epsilon$ varies.  

Then using elementary inequalities, we have that $v_\epsilon$ is increasing when $\epsilon$ decreases. But since $v_\epsilon$ is also upper bounded by inequality \eqref{boundedaway}, we deduce that $v_\epsilon$ admits a limit, say $v'$, when $\epsilon$ decreases to zero, and this limit still satisfies $v' < a^{-1}(c)$. Then it suffices to remark that the bivariate function:
$$
(\epsilon, v) \mapsto h_\epsilon(v)
$$
is continuous. Thus since for any $\epsilon$, we have $h_\epsilon(v_\epsilon)=0$, by letting $\epsilon$ go to zero, we have $h_0(v')=0$. That is, $v'$ is also a zero of equation \eqref{functiong0}. Then by Proposition 1 and inequality \eqref{boundedaway}, we conclude that $v'=v_0$. 

$ii)$ Let us now consider the alternative case where $a(1) \leq c$. By similar arguments as above, when $\epsilon$ goes to zero, $v_\epsilon$ is increasing and thus has a limit that is non larger than 1. Then we conclude that this limit is actually a zero of equation \eqref{functiong0}, hence the only possibility is that this limit is 1. 

%Function $g$ is such that~:
%
%$$
%g(0) = - \infty, g(1) = 1.
%$$

\section{Euler Discretization}
 
To avoid the numerical inversion of the integral equation (2.2) in Proposition 2, it is standard to approximate the continuous time solution of system (2.1), by considering the solution of its Euler time discretization. More precisely we fix some basic time unit $1/h$, say, where $h$ is interpreted as a time frequency, and consider the approximation~: $\tilde{x}_h (t) = x_h ([th]), \tilde{y}_h (t) = y_h ([th]), \tilde{z}_h(t) = z_h ([th]),$ where $[.]$ denotes the integer part, and the trivariate discrete time process $x_h (n), y_h (n), z_h (n)$ satisfies the system of recursive equations~:

\begin{equation}
  \left\{
      \begin{aligned}
        x_h(n) - x_h (n-1)&=&-\Frac{1}{h}& a [x_h (n-1)] y_h (n-1),\\
        y_h (n) - y_h (n-1)&=&\Frac{1}{h}& \{ a[x_h (n-1)] y_h (n-1) - cy_h (n-1)\},\\
        z_h (n) - z_h (n-1)&=&\Frac{1}{h}& c y_h (n-1).\\
      \end{aligned}
    \right.
\end{equation}

For ultra high frequency, that is if timestep $1/h$ tends to infinity, it is expected that~:

$$
\mbox{sup}_t |x(t) - \tilde{x}_h (t)|, \mbox{sup}_t |y(t) - \tilde{y}_h (t)|, \mbox{sup}_t |z(t) - \tilde{z}_h (t)|,
$$

\noindent tend to zero. This arises under sufficient Lipschitz conditions [Lakoba (2012)]. However, these conditions are not always satisfied in our nonlinear framework [see also Breto et al. (2009), Section 2.1 for a similar remark] when there is herd immunity and this standard approximation might no longer be used to simulate observed trajectories, or to derive numerically the long run properties of the model. {More precisely some Euler discretizations can have chaotic behaviours or spurious cycles that do not exist in the continuous time analogue [see Allen (1994)]}. We consider below different Euler discretizations of the SIR model and discuss their asymptotic behaviour.\vspace{1em}

The main result of this appendix concerns the herd immunity. Several Euler discretizations do not provide the same condition for herd immunity than their continuous counterpart, even if the timestep is very small. 
\subsection*{Example A.1~: Exact time discretization}\vspace{1em}

Let us denote $\tilde{x} (t;\theta)$ with $\theta = (x(0), y(0), c, a)$ the solution (2.2) of the SIR model and $\tilde{x}^{-1} (.;\theta)$ its inverse with respect to time. Since $t=(t-1) +1$, we have : $\tilde{x}^{-1} [x(t);\theta] = \tilde{x}^{-1} [x(t-1);\theta]+1, \forall t \in (0,\infty),$ or $x(n) = \tilde{x} [\tilde{x}^{-1} [x (n-1);\theta] + 1], \forall n$. This is the exact nonlinear difference equation of $x(n)$, with the same asymptotic properties than $x(t)$ by construction.\vspace{1em}

\subsection*{Example A.2~: Crude Euler discretization (a.17)}\vspace{1em}

Let us now consider the Euler discretization (a.17). Without loss of generality, we can set $h=1$, and omit index $h$ in system (a.17). We have the following properties~:

\vspace{1em}

\textbf{Lemma 1~:} If $0<a(x)<x$, for $x \in ]0,1], 0<c<1, x(0)>0, y(0)>0, z(0) = 1-x(0) - y(0) \geq 0,$ then,

i) The solution of the discrete time approximation is such that~: 

$x(n) > 0, y(n)> 0, z(n)\geq 0, \forall n$.

ii) The sequence $x(n)$ is decreasing with a limit $x(\infty)$.

iii) The sequence $z(n)$ is increasing with a limit $z(\infty)$.

iv) The sequence $y(n)$ tends to $y(\infty) = 0$.\vspace{1em}

\textbf{ Proof :} i) This is proved by recursion, assuming $x(n-1) > 0, y (n-1) > 0, z(n-1) \geq 0$, with by construction: $x(n-1) + y(n-1)+z(n-1) = 1.$\vspace{1em}

This is a consequence of~:

$x(n) = x(n-1) \{1-\Frac{a[x(n-1)]}{x(n-1)} y (n-1)\}$, with $\Frac{a[x(n-1)]}{x(n-1)} < 1$ and $y(n-1) <1,$ of $y(n) = y(n-1) \{ 1-c + a[x(n-1)]\},$ with $y(n-1) >0,1-c>0, a [x(n-1)]>0,$ and of $z(n) = z(n-1) + cy(n-1).$\vspace{1em}

ii) Indeed~:

$x(n) = x(n-1) - a[x(n-1)] y (n-1)$, with $a [x(n-1)] y (n-1) > 0.$

iii) The result is straightforward~:

iv) By considering the equation $z(n) = z(n-1) + cy (n-1)$, when $n$ tends to infinity, we get $y (\infty) = 0$.

\QED\vspace{1em}

\textbf{Remark a.1~:} When the timestep is introduced, the condition on function $a$ becomes $\frac{1}{h}a(x) < x.$\vspace{1em}

\textbf{Remark a.2~:} When $a(x) = \alpha x^\beta,$ the condition $a(x) <x, x \in (0,1)$, is equivalent to $\beta \geq 1$, that is the same condition for $\beta$ as in the continuous SIR model.\vspace{1em}

\textbf{Proposition A.1~:} Under the assumptions of Lemma 1, the limit $x(\infty)$ is strictly positive. \vspace{1em}

\textbf{Proof~:} We have~:

$$
x(n) = x(n-1) - a [x(n-1)] y(n-1) > x(n-1) [1-y(n-1)].
$$

Therefore,

$$
x(\infty) > x(0) \exp \left[ \Sum^\infty_{j=0} \log [1-y (j)]\right].
$$

The lower bound is strictly positive, if the negative series $\log [1-y(n)]$ is summable, or, equivalently since $y(n)$ tends to zero, if the series $y(n)$ is summable. This summability condition is satisfied since $\Sum^\infty_{n=0} y(n) = z(\infty) - z(0) \leq 1.$

\QED

Next under an additional assumption on transmission function $a$, we get the analogue of Proposition 2 on the evolution of $y(n)$ and the possibility of a peak of the epidemic.\vspace{1em}

\textbf{Proposition A.2~:} Under the assumptions of Lemma 1 and if $a$ is an increasing function, then there exists a value $n^* \geq 0$ such that the sequence $y(n)$ is positive, strictly decreasing after $n^*$.\vspace{1em}

\textbf{Proof~:} We have~: $y(n) - y(n-1) = y(n-1) [-c + a [x(n-1)]]$. Three cases can be distinguished.

i) If $-c+a[x(0)] > 1 \Longleftrightarrow \Frac{a[x(0)]}{c} < 1$, we have~: $-c + a [x(n-1)] < -c+a [x(0)], \forall n$, and the sequence $y(n)$ is positive strictly decreasing.\vspace{1em}

ii) If $-c + a[x(0)]>1$, and if $-c+a[x(0)]>1, \forall n$, then the sequence $y(n)$ would be strictly increasing, and then ~: $z(n) = z(0) + y(0) + \ldots + y(n-1)$  would be larger than $z(0) + n y(0),$ and would tend to infinity. This contradicts the fact that $z(n) \leq 1$ (by Lemma 1).\vspace{1em}

iii) Therefore, if $-c + a[x(0)]> 1$, there exists a value $n^* > 0$ such that $y(n)$ increases before $n^*$ and decreases after $n^*$.

\QED

\section{Frailty model}
A popular survival model for the mover-stayer phenomenon is the (proportional) frailty, or unobserved heterogeneity model\footnote{See Farrington and Whitaker (2003) for an application of frailty models to epidemiology. }. Example 2 above allows also for a similar interpretation in the special case where $\beta>1$. Let us assume that for each susceptible individual, given  a static, individual-specific unobserved heterogeneity random variable $\theta$, its conditional instantaneous hazard function of being infected is $\theta y(t)$. Then its conditional survival function given $\theta$ is:
$$
S(t|\theta)=\exp\Big(-\theta  \int_0^t  y(s)\mathrm{d}s\Big)
$$
and the unconditional survival function is, by definition, $x(t)$. That is:
\begin{equation}
	\label{unconditionalsurvivor}
	x(t)=\mathbb{E}[S(t|\theta)]=\mathbb{E}\Big[\exp\Big(-\theta \int_0^t  y(s) \mathrm{d}s \Big) \Big]=\mathcal{L}\Big( \int_0^t  y(s)\mathrm{d}s\Big),
\end{equation}
where $\mathcal{L}(u)=\mathbb{E}[e^{-u \theta}]$, $u>0$, is the Laplace transform of random variable $\theta$ taken at $u$. Let us for instance assume that $\theta$ follows the gamma distribution with rate parameter $\kappa$ and degrees of freedom $\delta$. Then we have $x(t)= \frac{1}{(1+\frac{1}{\kappa}\int_0^t  y(s) \mathrm{d}s )^\delta}$. Then by differentiating, we get:
\begin{align*}
	\frac{\mathrm{d}x(t)}{\mathrm{d}t}=-\frac{\delta}{\kappa}  [x(t)]^{\frac{\delta+1}{\delta}}y(t).
\end{align*}
Hence we recover the model in Example 2, with $\alpha=\frac{\delta}{\kappa}$ and $\beta=\frac{\delta+1}{\delta}>1$. In particular, while the conditional hazard $\theta y(t)$ depends only on $y(t)$ but is stochastic because of $\theta$, the unconditional hazard is equal to $\frac{\delta}{\kappa} y(t) [x(t)]^{\frac{1}{\delta}} $ is also decreasing in $x(t)$ because of the mover-stayer phenomenon. 

The gamma distribution assumption is often regarded the ``reference" distribution in survival analysis due to both its mathematical tractability and theoretical, asymptotic properties [see Abbring and van den Berg (2007)]. Alternatively, other functional forms for the distribution $\theta$ are equally possible and we would get, in general, $$
x'(t)=-y(t)\mathcal{L}' \circ \mathcal{L}^{(-1)} [x(t)],
$$
where $\mathcal{L}^{(-1)}$ is the inverse function of the Laplace transform. In particular, so long as this inverse is tractable (and hence the derivative of this latter), we have $a(x)=\mathcal{L}' \circ \mathcal{L}^{(-1)}(x)=\frac{1}{\frac{\mathrm{d}}{\mathrm{d}x}\mathcal{L}^{(-1)}(x)}$ will be tractable. 
%by another distribution on $\mathbb{R}_{>0}$, then we will get other functional forms for $a(\cdot)$. 
\vspace{1em}

For instance, another population frailty distribution in the survival analysis literature is the $\alpha-$stable distribution, see e.g. Hougaard (1986). More precisely, if $\theta$ follows the positive  with $\mathcal{L}(u)=e^{-u^{\alpha_s}}$ where the stability parameter $\alpha_s$ is between 0 and 1, then we get:
$$
\frac{\mathrm{d}x(t)}{\mathrm{d}t}=-\alpha_s x(t) y(t)[- \ln x(t)]^{\frac{\alpha_s-1}{\alpha_s}}.
$$
Thus we have $a(x)=\alpha_s x [- \ln x]^{\frac{\alpha_s-1}{\alpha_s}}$. This function is well defined on $[0,1]$, but $a(\cdot)$ is not differentiable\footnote{Indeed, the $\alpha-$stable distribution has an infinite mean. } at $x=1$. In other words, under this model, the initial progression of the epidemic is extremely quick.


\newpage
\begin{thebibliography}{}
	\bibitem{} Abbring J., and G. Van Den Berg (2007) ~: ``The Unobserved Heterogeneity Distribution in Duration Analysis." Biometrika, 94(1), 87-99.

\bibitem{} Alipoor, A., and O., Boldea (2020)~: "The Role of Elementary Schools in SARS-CoV-2 Transmission", Tilburg University DP. 

\bibitem{} Allen, L. (1994)~: "Some Discrete-Time SI, SIR and SIS Epidemic Models", Mathematical Biosciences, 184, 83-105.

\bibitem{} Anderson, R., and R., May (1991)~: "Infectious Diseases of Humans : Dynamics and Controls", Oxford University Press.

\bibitem{} Bailey, N. (1980) : "Spatial Models in the Epidemiology of Infectious Deceases", in Biological Growth and Spread, Lecture Notes in Biomath. 38, 233-261.


\bibitem{} Bernoulli, D. (1760)~: "Essai d'une nouvelle analyse de la mortalit\'e caus\'ee par la petite v\'erole et des avantages de l'inoculation pour la pr\'evenir", Mem. Math. Phys. Acad. Roy-Sci., Paris.

    \bibitem{} Blumen, I., Kogan, M., and P., McCarthy (1955)~: "The Industrial Mobility of Labor as a Probability Process", in Vol VI of Cornell Studies of Industrial and Labor Relations, Cornell Univ., Ithaca.

\bibitem{}Boatto, S., Bonnet, C., Cazelles, B., and F., Mazenc (2018)~: "SIR Model with Time Dependent Infectivity Parameter :Approximating the Epidemic Attractor and the Importance of the Initial Phase," CentraleSupelec DP. 

    \bibitem{} Bohner, M., and A., Peterson (eds) (2003)~: "Advances in Dynamic Equations on Time Scales", Birkhauser, Boston.

        \bibitem{} Bohner, M., Streipert, S., and D., Torres (2019)~: "Exact Solution of a Dynamic SIR Model", Nonlinear Analysis : Hybrid System, 32, 228-238.

\bibitem{} Brauer, F., and C., Castillo-Chavez (2001)~: "Mathematical Models in Population Biology and Epidemiology", Texts in Applied Mathematics, 40, Springer Verlag, New-York.

\bibitem{} Breto,C., He, D., Ionides, E. and A., King (2009): "Time Series Analysis via Mechanistic Models", Annals of Applied Statistics, 3, 319-348.

\bibitem{} Butcher, J. (2003) : "Numerical Methods for Ordinary Differential Equations", Wiley, New-York.

\bibitem{}Chowell, G., Sattenspiel, L., Bansal, S., and C., Viboud (2016)~: "Mathematical Models to Characterize Early Epidemic Growth: A Review", Physics of Life reviews, 18, 66-97.

    \bibitem{} Clement, E., Gourieroux, C., and A., Monfort (2000)~: "Econometric Specification of the Risk-Neutral Valuation Model", Journal of Econometrics, 94, 117-143.
    
    \bibitem{} Cori, A., Ferguson, N., Fraser, C., and S., Cauchemez (2013)~: "A New Framework and Software to Estimate Time Varying Reproduction Numbers During Epidemics", American Journal of Epidemiology, 178, 1505-1512.
    
    \bibitem{}Cox, D. (1972): "Regression Models and Life Tables." Journal of the Royal Statistical Society: Series B (Methodological) 34, 187-202.

\bibitem{} Cox, J. (1996)~: "The Constant Elasticity of Variance Option Pricing Model", Journal of Portfolio Management, 22, 15-17.

\bibitem{} Cox, J., Ingersoll, J., and S., Ross (1985)~: "A Theory of the Term Structure of Interest Rates", Econometrica, 53, 385-407.

   % \bibitem{} Das, P., Mukherjee, D., and A., Sarkar (2011)~: "Study of an SI Epidemic Model with Nonlinear Incidence Rate : Discrete and Stochastic Version", Applied Mathematics and Computation, 218, 2509-2515.

\bibitem{} Djogbenou, A., Gourieroux, C., Jasiak, J., Rilstone, P., and M., Bandehali (2020)~: "Transition Model for Coronavirus Management", COVIDECON, 35, 176-219.

\bibitem{} Dureau,J., Kalogeropoulos, K., and M., Buguelin (2013)~: "Capturing the Time Drivers of an Epidemic Using Stochastic Dynamical Systems", Biostatistics, 14, 541-555.

        \bibitem{} El Koufi, A., Adnani, J., Bennar, A., and N., Yousfi (2019)~: "Analysis of a Stochastic SIR Model with Vaccination and Nonlinear Incidence Rate", International Journal of Differential Equation.

\bibitem{} Elliott, S., Gourieroux, C., and N., Meddahi (2020)~: "Uncertainty on the Reproduction Ratio in the SIR Model", University of Toronto DP.

\bibitem{} Fan, K., Zhang, Y., Gao, S., and X., Wei (2017)~: "A Class of Stochastic Delayed SIR Epidemic Models with Generalized Nonlinear Incidence Rate and Temporary Immunity", Physica A., 481, 176-190.

\bibitem{}Farrington, P., and H. Whitaker (2003)~: "Estimation of Effective Reproduction Numbers for Infectious Diseases using Serological Survey Data." Biostatistics 4.4, 621-632.


    \bibitem{} Feller, W. (1940)~: "On the Logistic Law of Growth and its Empirical Verifications in Biology", Acta Biotheoretica, 5, 51-66.

\bibitem{} Fraser, C. (2007)~: "Estimating Individual and Household Reproduction Numbers in An Emerging Epidemic", PLOS one, 2(8), e758.

       \bibitem{} Goodman, L. (1965)~: "Statistical Methods for the Mover-Stayer Model", JASA, 56, 841-868.

        \bibitem{} Gourieroux, C., and J., Jasiak (2020a)~: "Time Varying Markov Process with Partially Observed Aggregate Data : An Application to Coronavirus", forthcoming Journal of Econometrics.

\bibitem{} Gourieroux, C., and J., Jasiak (2020b)~: "Analysis of Virus Transmission:  A Transition Model Representation of Stochastic Epidemiological Models", forthcoming Annals of Economics and Statistics.
   %    \bibitem{} Guan, L., Hengartner, N., and J., Hyman (2018)~: "Investigation on the Uncertainty of the Parameters in the Simple Susceptible-Infectious Recovered Model", Preprint, Tulane Univ., New-Orleans.

 \bibitem{} Guan, L., Li, D., Wang, K., and K., Zhao (2019)~: "On a Class of Nonlocal SIR Models", Journal of Mathematical Biology, 78, 1581-1604.

\bibitem{} Harko, T., Lobo, F., and M., Mak (2014)~: "Analytical Solutions of the Riccati Equation with Coefficients Satisfying Integral or Differential Conditions with Arbitrary Functions", Universal Journal of Applied Mathematics, 2, 109-118.

\bibitem{} Harko, T., Lobo, F., and M., Mak (2016)~: "Exact Analytical Solutions of the Susceptible-Infected-Recovered (SIR) Epidemic Model and of the SIR Model with Equal Death and Birth Rates", Applied. Math. and Computation, 236, 184-194.

\bibitem{} Heffernan, M., Robert S., and L. Wahl(2005)~ : "Perspectives on the Basic Reproductive Ratio." Journal of the Royal Society Interface 2.4, 281-293. 

    \bibitem{} Hethcote, H. (2000)~: "The Mathematics of Infectious Diseases", SIAM Review, 42, 599-653.
    
\bibitem{} Hougaard, P. (1986)~: ``Survival models for Heterogeneous Populations Derived From Stable Distributions." Biometrika 73(2), 387-396.

\bibitem{} Jiang, D., Yu, J., Ji, C., and N., Shi (2011)~: "Asymptotic Behaviour of Global Positive Solution to a Stochastic SIR Model", Math. Comput. Modell., 54, 221-232.

%\bibitem{}Katriel, G. (2012)~: "The Size of Epidemics in Populations with Heterogeneous Susceptibility", Journal of Mathematical Biology, 65, 237-262.

\bibitem{} Kahi, D., Kebraeb, E., Lopez, S., and J., France (2003)~: "An Evaluation of Different Growth Functions for Describing the Profile of Live Weight with Age in Meat and Egg Strains of Chicken", Poultry Science, 82, 1536-1543.

\bibitem{} Kermack, W., and A., McKendrick (1927)~: "Contribution to the Mathematical Theory of Epidemics", Proc. Roy. Soc., London, 115,700-721.

      %  \bibitem{} Khasminskii, R. (2011)~: "Stochastic Stability of Differential Equations", Springer Verlag, Berlin.

\bibitem{} Korobeinikov, A., and G., Wake (2002)~: "Lyapunov Functions and Global Stability for SIR, SIRS and SIS Epidemiological Models", Applied Math Letter, 15, 955-960.

\bibitem{} Lakoba, T. (2012)~: "Simple Euler Method and its Modifications", Univ. of Vermont.

\bibitem{} Li, J., Blakeley, D. and R., Smith (2011)~: ``The Failure of $\mathcal{R}_0$", Computational and Mathematical Methods in Medicine, 2011.

 \bibitem{} Li, D., Cui, J., Liu, M., and S., Liu (2015)~: "The Evolutionary Dynamics of Stochastic Epidemic Model with Nonlinear Incidence Rate", Bull. Math. Biol, 77, 1705-1743.

\bibitem{}Lin, F., Peng, L., Xie, J., and J., Yang (2018)~: ``Stochastic Distortion and its Transformed Copula," Insurance: Mathematics and Economics, 79, 148-166.



            \bibitem{} Ma, J., and D., Earn (2006)~: "Generality of the Final Size Formula for an Epidemic of a Newly Invading Infectious Disease", Bulletin of Mathematical Biology, 68, 679-702.

            \bibitem{} Madan, D., Carr, P., and E., Chang (1998)~: "The Variance-Gamma Process and Option Pricing", European Finance Review, 2, 79-105.

\bibitem{} Miller, J. (2012)~: "A Note on the Derivation of Epidemic Final Sizes",
Bulletin of Mathematical Biology, 74, Section 4.1.

                \bibitem{} Niepelt, D., and M., Gonzales-Eiras (2020)~: "Tractable Epidemiological Models for Economic Analysis", CEPR VOX eu.

\bibitem{} Novozhilov, A. (2008): "On the Spread of Epidemics in a Closed Heterogeneous Population", Mathematical Biosciences, 215, 177-185.

\bibitem{} O'Regan, M., Kelly, T., Korobeinikov, A.,  O'Callaghan, M., and A., Pokrovskii (2010)~: "Lyapunov Functions for SIR and SIRS Epidemic Models", Appl. Math. Letter, 23, 446-448.

\bibitem{} Patie, P., and M., Savov (2013). "Exponential functional of L\'evy processes: Generalized Weierstrass products and Wiener-Hopf factorization", Comptes Rendus Mathematique, 351(9-10), 393-396.

\bibitem{} Revuz, D., and M., Yor (2001)~: "Continuous Martingales and Brownian Motion", Third Edition, Springer, Berlin.

    \bibitem{} Richards, F. (1959)~: "A Flexible Growth Function for Empirical Use", J. Exp. Bot., 10, 290-301.

        \bibitem{} Rifhat, R., Wang, L., and Z., Teng (2017)~: "Dynamics for a Class of Stochastic SIS Epidemic Models with Nonlinear Incidence and Periodic Coefficients", Physica, A., 481, 176-190.
        
        \bibitem{} Roy, M., and M., Pascual (2006):" On Replicating Network Heterogeneity in the Incidence Rate of Simple Epidemic Models", Ecological Complexity, 3, 80-90.

\bibitem{}Stroud, P., Sydoriak, S., Riese, J., Smith, J., Mniszewski, S., and P., Romero (2006): "Semi-Empirical Power-Law Scaling of New Infection Rate to Model Epidemic Dynamics with Inhomogeneous Mixing", Mathematical Biosciences, 203(2), 301-318.

\bibitem{}Tildesley, J., and J., Keeling (2009)~: ``Is $\mathcal{R}_0$ a Good Predictor of Final Epidemic Size: Foot-and-Mouth Disease in the UK", Journal of Theoretical Biology, 258(4), 623-629.

        \bibitem{} Williams, M., Mazilu, I., and A., Mazilu (2012)~: "Stochastic Epidemic-Type Model with Enhanced Connectivity : Exact Solution", Journal of Statistical Mechanics : Theory and Experiment, 2012.
        
        \bibitem{} Wilson, E., and J., Worcester (1945): "The Law of Mass Action in Epidemiology", Proc. Natl. Acad. Sci. U.S.A., 31, 24-34.

            \bibitem{} Xu, C., and X., Li (2018)~: "The Threshold of a Stochastic Delayed SIRS Epidemic Model with Temporary Immunity and Vaccination", Chaos, Solitons and Fractals, 111, 227-234.

        \bibitem{} Zhang, I., and Z., Ma (2003)~: "Global Dynamics of a SEIR Epidemic Model with Saturating Contact Rate", Mathematical Biosciences, 185, 15-35.

\end{thebibliography}
\end{document}